\documentclass[12pt,a4paper]{article}
\usepackage[margin=2cm]{geometry}
\usepackage{caption}
\usepackage{amsmath,amsthm,amssymb,bm,hhline,psfrag,comment}
\usepackage{subfigure}
\usepackage[hyphens]{url}
\usepackage{color}
\usepackage{proof}
\usepackage[dvipdfmx]{graphicx}
\usepackage[authoryear,round]{natbib}

\def\hsymbu#1{\smash{\lower1.7ex\hbox{\huge$#1$}}}

\usepackage{indentfirst}
\usepackage{algorithm,algorithmic}
\usepackage{appendix}
\usepackage[normalem]{ulem}
\usepackage{nomencl}

\theoremstyle{plain}
\newtheorem{lemma}{Lemma}[section]
\theoremstyle{remark}

\theoremstyle{example}

\theoremstyle{lemma}
\newtheorem{prop}{Proposition}[section]
\newtheorem{Assumption}{Assumption}[section]
\theoremstyle{Theorem}
\newtheorem{Theorem}{Theorem}[section]

\newcommand{\argmin}{\mathop{\rm argmin}}
\newcommand{\argmax}{\mathop{\rm argmax}}
\newcommand{\as}{{\ \xrightarrow{\rm a.s.}} \ }

\begin{document}
{%\baselineskip 8mm
\begin{center}
\textbf{\Large Hierarchical clustered multiclass discriminant analysis via cross-validation
%with limited external information
}
\end{center}
\begin{center}
\large {Kei Hirose$^{1,2}$, Kanta Miura$^{1}$ and Atori Koie$^{3}$
}
\end{center}

\begin{flushleft}
{\footnotesize
$^1$ Institute of Mathematics for Industry, Kyushu University, 744 Motooka, Nishi-ku, Fukuoka 819-0395, Japan \\
$^2$ RIKEN Center for Advanced Intelligence Project, 1-4-1 Nihonbashi, Chuo-ku, Tokyo 103-0027, Japan \\
%$^3$ Graduate School of Mathematics, Kyushu University, 744 Motooka, Nishi-ku, Fukuoka 819-0395, Japan \\
$^3$ Nissan Motor Co., Ltd., 1-1, Morinosatoaoyama, Atsugi, Kanagawa 243-0123, Japan \\
}
 {\small E-mail: {\it hirose@imi.kyushu-u.ac.jp} (K.H.); {\it kantamiura903@gmail.com} (K.M.)}
\end{flushleft}

\vspace{1.5mm}

\begin{abstract}
Linear discriminant analysis (LDA) is a well-known method for multiclass classification and dimensionality reduction.  However, in general, ordinary LDA does not achieve high prediction accuracy when observations in some classes are difficult to be classified.  This study proposes a novel cluster-based LDA method that significantly improves the prediction accuracy.  We adopt hierarchical clustering, and the dissimilarity measure of two clusters is defined by the cross-validation (CV) value. Therefore, clusters are constructed such that the misclassification error rate is minimized.  Our approach involves a heavy computational load because the CV value must be computed at each step of the hierarchical clustering algorithm.  To address this issue, we develop a regression formulation for LDA and construct an efficient algorithm that computes an approximate value of the CV.  The performance of the proposed method is investigated by applying it to both artificial and real datasets.  Our proposed method provides high prediction accuracy with fast computation from both numerical and theoretical viewpoints.
\end{abstract}
 \noindent {\bf Key Words}:  Cross-validation; Linear discriminant analysis; Hierarchical clustering; Regression formulation

\section{Introduction}
\label{intro}
Linear discriminant analysis (LDA) is a well-known technique for multiclass classification problems based on the covariance structure of each class \citep{fisher1936use,rao1948utilization,fukunaga2013introduction}.  Furthermore, it is a dimension reduction tool for understanding the positional relationship among multiple classes in high-dimensional data. The implementation of LDA is easy because it results in a generalized eigenvalue problem.  Various extensions of LDA have been proposed in the literature.  For example, LDA can be extended to nonlinear analysis by kernel method (Fisher discriminant analysis; \citealp{mika1999fisher,baudat2000generalized}).  Sparse multiclass LDA  \citep{clemmensen2011sparse,shao2011sparse,witten2011penalized,safo2016general} is a useful technique for interpreting  the classification results for high-dimensional data. \citet{Wu:2017ia} introduced a hybrid version of LDA and deep neural networks, and applied it to person re-identification.  

In practice, observations in some classes are often easy to be classified, whereas those in other classes are difficult to be classified.  Figure \ref{fig:ldamouse} shows 2D projections of $36$-dimensional data points obtained by LDA of mouse consomic strain data \citep{takada2008mouse}.  The number of classes is 30.   Most species are similar; however, some species have different characteristics.   For example, classification of an input whose label is either ``MSM" or ``B6-11MSM" seems easy.  Meanwhile, inputs with other labels may be difficult to classify correctly; the data points in similar classes are not sufficiently separated in the two-dimensional space.   Therefore, ordinary LDA may lead to a high misclassification error rate.  
\begin{figure}
\centering
    \subfigure[2D plot of projected data points.]{\includegraphics[width=80mm, bb=0 0 566 425]{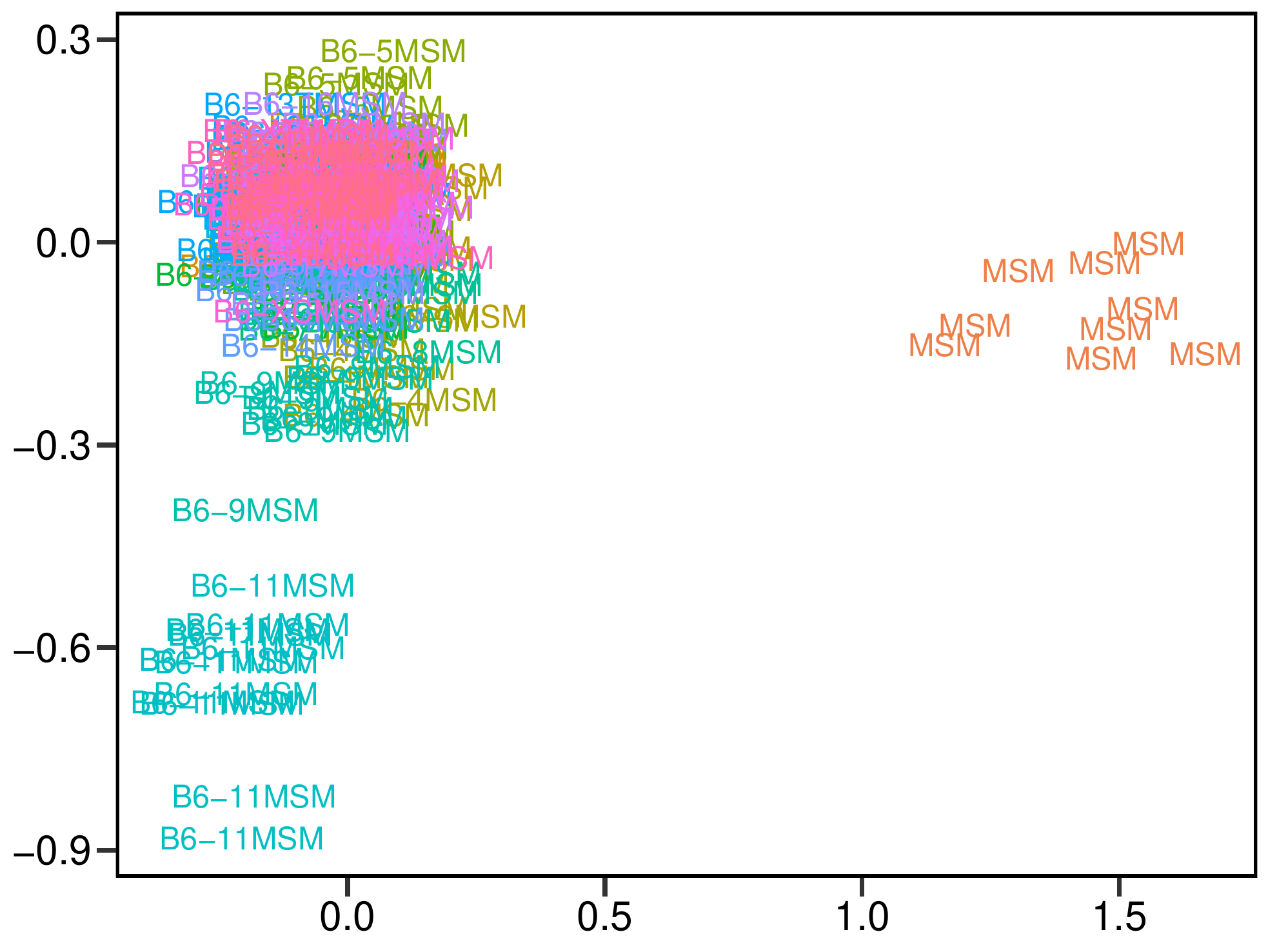}\label{fig:ldamouse}}
        \hspace{0.5cm}
    \subfigure[Dendrogram of Ward's hierarchical clustering. ]{\includegraphics[width=80mm, bb=0 0 566 425]{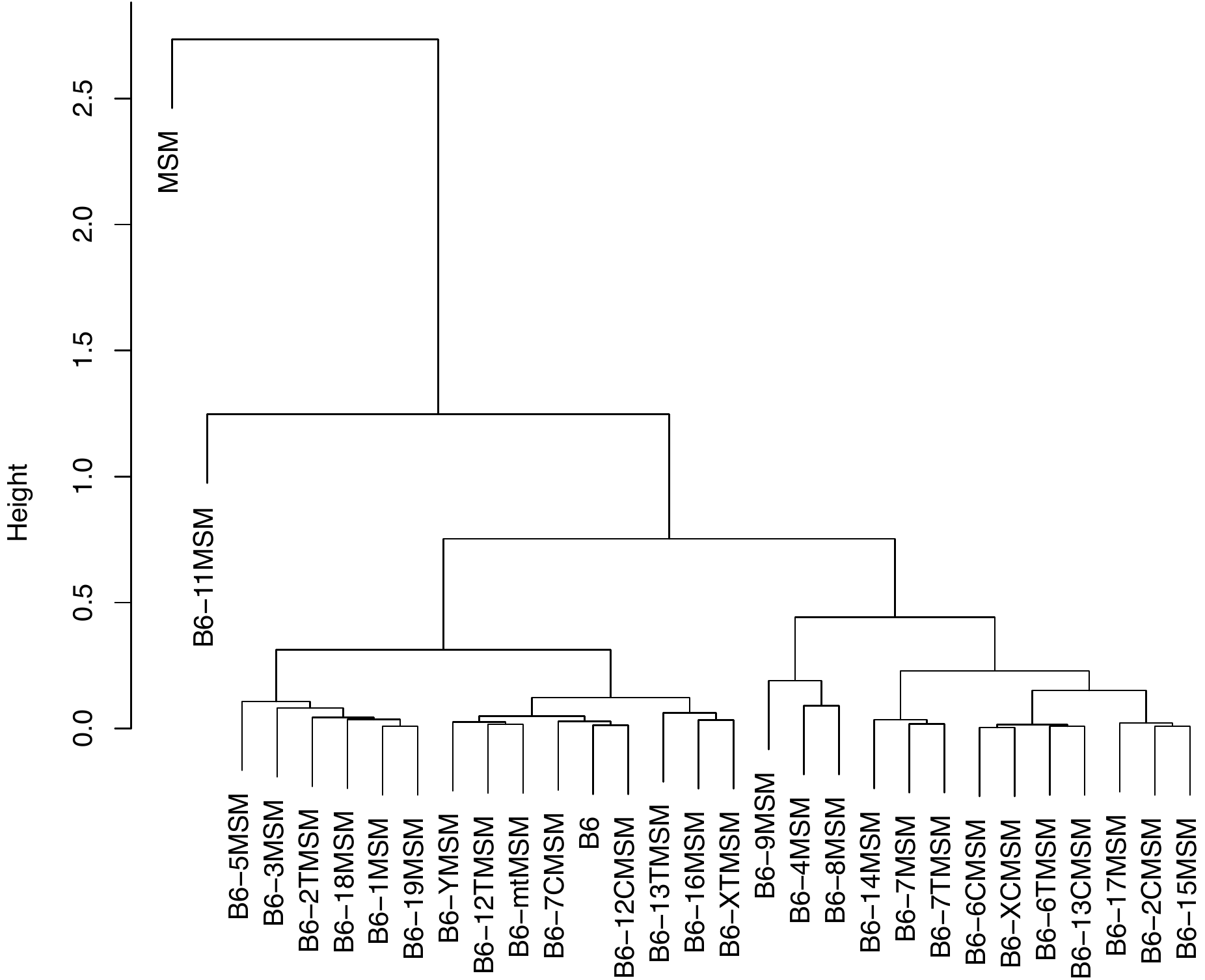}\label{fig:ldamouse_cluster} }
\caption{Two-dimensional projections of data points obtained by (a) LDA and (b) the dendrogram of Ward's hierarchical clustering for the mean value of each class.}
\end{figure}

To address this issue, one can adopt the following two-stage procedure. In the first stage, we calculate a representative point in each class (e.g., mean value) and apply standard cluster analysis (Ward's method, $k$-means clustering, etc.) to the representative points.  In the second stage, we separately apply LDA to each cluster.  With the two-stage procedure, each cluster has different classification boundaries, leading to more flexible boundaries compared to ordinary LDA. Thus, the misclassification error rate can be reduced.   Figure \ref{fig:ldamouse_cluster} shows the dendrogram of Ward's method for the mean vectors of the classes for the mouse consomic data.  The result suggests that three clusters can be constructed: ``MSM", ``B5-11MSM", and ``other 28 classes".  Thus, one can create the following classification rule: when a new input is observed, we may first classify one of these three clusters.  If the input is classified into ``other 28 classes", we classify one of the 28 classes.  A similar two-stage procedure based on linear and nonlinear discriminant functions has been proposed by \cite{Huang:2013is}.     

However, improvement of the prediction accuracy is not guaranteed with ordinary clustering algorithms, such as Ward's method.  In general, cluster analysis and classification problem have different purposes; cluster analysis is used to find collections of classes on the basis of the similarity of two classes (e.g., Euclidean distance and Mahalanobis distance), whereas LDA is conducted to construct the classification rule.  Combining two different methods does not usually guarantee minimization of the overall error rate (e.g., \citealp{yamamoto2014functional,kawano2015sparse,Kawano:2018hf}).  Therefore, it is crucial to develop a clustering technique that aims to minimize the error rate.

In this study, we propose a novel cluster-based LDA method to minimize the misclassification error rate.  We adopt hierarchical clustering to construct clusters of classes.  A crucial aspect is how to determine the clusters.   Ordinary hierarchical clustering merges clusters on the basis of a dissimilarity measure between two clusters. However, the dissimilarity measure does not aim to decrease the error rate. This study uses the leave-one-out cross-validation (CV) value because it is an unbiased estimator of the misclassification error rate \citep{lachenbruch1967almost}.  The CV value is computed for the following two-stage classification rule. In the first stage, we adopt LDA to create a classification rule that allocates the input to some cluster of classes.  In the second stage, we again apply LDA to each cluster.  The clusters are constructed such that the CV error is minimized.

The proposed approach involves a heavy computational load because the CV value must be computed at each step of the hierarchical clustering algorithm.  Therefore, it is necessary to construct an efficient algorithm to compute the CV value.  One can develop a regression formulation and use an efficient algorithm for ridge regression (see, e.g., \citealp{konishi2008information}). \cite{Cawley:2003ha} introduced an efficient algorithm using a regression formulation given by \citet{Xu:2001vi}; however, their method can be applied to only two-class LDA.  \cite{ye2007least} proposed a multivariate regression formulation for multiclass LDA; however, his method can be used only when the number of dimensions of the projected spaces, say $D$, is fixed.  In other words, we cannot select the value of $D$.  In practice, an analyst often needs to determine $D$ for interpretation purposes (e.g., $D=2$ for visualization).  Moreover, the prediction accuracy is strongly dependent on $D$.  Thus, \citeauthor{ye2007least}'s \citeyearpar{ye2007least} approach is limited.  In this study, we develop a novel regression formulation for multiclass classification in which a user can determine $D$.  Accordingly, we derive an efficient algorithm that computes an approximate CV value.  A theoretical justification for the approximation is also presented.

Monte Carlo simulation is conducted to investigate the performance of the proposed procedure.  The usefulness of the proposed procedure is illustrated through the analysis of mouse consomic strain data.  We provide an {\tt R} package {\tt hclda} for implementing our algorithm; it is available at \url{https://github.com/keihirose/hclda}.

The remainder of this article is organized as follows. Section \ref{sec:LDA} briefly reviews multiclass LDA. Section \ref{sec:proposed} formulates the proposed algorithm on the basis of a two-stage procedure with LDA.  Section \ref{sec:algorithm} derives an efficient algorithm for computing the CV value in multiclass LDA.  Section \ref{sec:numerical} discusses the effectiveness of our procedure on the basis of Monte Carlo simulation.  Section \ref{sec:realdata} presents real data analysis . Finally, Section \ref{sec:conclusion} concludes the paper. Some technical proofs are deferred to the appendices.

\section{Linear discriminant analysis}
\label{sec:LDA}
Suppose that we have $n$ observations $\{(y_{i}, \bm{x}_{i})$ $(i=1,\ldots, n)\}$ with respect to class labels $y_{i}\in \{1, \ldots, J\}$ and predictor vectors $\bm{x}_{i}=(x_{i1}, \ldots, x_{ip})^{\mathsf{T}}$, where $J$ is the number of classes and $\bm{A}^{\mathsf{T}}$ denotes the transpose of $\bm{A}$.  Let $\mathcal{G} = \{G_1,...,G_J\}$, where $G_j$ is the subset of $\{\bm{x}_1,\dots,\bm{x}_n\}$ that belongs to the $j$th class; when $y_i=j$, we have $\bm{x}_i \in G_j$.        

Let $\bar{\bm{x}}_{j}$ be the sample mean vector in class $j$ and $\bar{\bm{x}}$ be the sample mean vector:
\begin{align}
&\bar{\bm{x}}_{j}=\frac{1}{n_{j}} \sum_{i: y_{i}=j}^{} \bm{x}_{i} ,\nonumber  \\ 
&\bar{\bm{x}}=\frac{1}{n} \sum_{j=1}^{J} n_{j}\bar{\bm{x}}_{j} , \nonumber
\end{align}
where $n_{j}$ is the number of observations in class $j$. 
Let $\bm{S}_{B}$ be the between-class covariance matrix and $\bm{S}_{W}$ be the within-class covariance matrix:
\begin{align}
&\bm{S}_{B}=\frac{1}{n}\sum_{j=1}^{J} n_{j} (\bar{\bm{x}}_{j}-\bar{\bm{x}})(\bar{\bm{x}}_{j}-\bar{\bm{x}})^{\mathsf{T}},\nonumber\\ 
&\bm{S}_{W}=\frac{1}{n}\sum_{j=1}^{J}\sum_{i:y_{i}=j}^{} (\bm{x}_{i}-\bar{\bm{x}}_{j})(\bm{x}_{i}-\bar{\bm{x}}_{j})^{\mathsf{T}}. \nonumber
\end{align}

Through LDA, a $p$-dimensional predictor $\bm{x}$ is transformed into a $D$-dimensional space ($D<\min(J-1,p)$) by linear transformation matrix $\bm{A}$, which is obtained by maximizing the ratio of the between-class covariance matrix and the within-class covariance matrix: 
\begin{align}
\hat{\bm{T}}=\argmax_{\bm{T}  \in \mathbb{R}^{p\times D}} \mathrm{tr}[(\bm{T}^{\mathsf{T}}\bm{S}_{W}\bm{T})^{-1}(\bm{T}^{\mathsf{T}}\bm{S}_{B}\bm{T} )] \ \ \mbox{subject to}\ \bm{T}^{\mathsf{T}}\bm{S}_{W}\bm{T}=\bm{I}_{D}. \label{eq:5}
\end{align}
The optimization problem given by Eq. $(\ref{eq:5})$ results in a generalized eigenvalue problem; the $d$-th column of $\hat{\bm{T}}$ $(d=1,\ldots, D)$ is given by the eigenvector corresponding to the $d$-th largest eigenvalue of matrix $\bm{S}_{W}^{-1}\bm{S}_{B}$.

 The minimum eigenvalue of $\bm{S}_W$ is often small when $p$ is large.  Indeed, the inverse matrix of $\bm{S}_W$ does not exist when $p > n$.  In such cases, one can use the ridge penalization $\bm{S}_{W,\delta} :=\bm{S}_W + \frac{\delta}{n} \bm{I}$, where $\delta \geq 0$ is a regularization parameter.  The optimization problem is then expressed as
\begin{align}
\hat{\bm{T}}_{\delta}=\argmax_{\bm{T}  \in \mathbb{R}^{p\times D}} \mathrm{tr}[(\bm{T}^{\mathsf{T}}\bm{S}_{W,\delta}\bm{T})^{-1}(\bm{T}^{\mathsf{T}}\bm{S}_{B}\bm{T} )] \ \ \mbox{subject to}\ \bm{T}^{\mathsf{T}}\bm{S}_{W,\delta}\bm{T}=\bm{I}_{D}.\nonumber
\end{align}

Now, we consider a multiclass classification problem.   Given $\bm{x}\in \mathbb{R}^{p}$, the $p$-dimensional observation $\bm{x}$ is transformed as $\bm{z}:=\hat{\bm{T}}_{\delta}^{\mathsf{T}} \bm{x}$.  The class of $\bm{x}$, say $G(\bm{x})$, is assigned using a Euclidean distance in the transformed space \citep{Witten:2011kc}:
\begin{align}
G(\bm{x})=\argmin_{j\in \{1,\ldots, J\}} \left(\bm{z}-\bar{\bm{z}}_{j}\right)^{\mathsf{T}} \left(\bm{z}-\bar{\bm{z}}_{j}\right), \nonumber
\end{align}
where $ \bar{\bm{z}}_{j}=\hat{\bm{T}}_{\delta}\bar{\bm{x}}_j$.

\section{Proposed method}
\label{sec:proposed}

As described in the Introduction, observations in some classes are easy to be classified whereas those in others are difficult to be classified.  To address this issue, we define a set of classes whose observations may not be easy to be classified with ordinary LDA. We refer to a set of classes as {\it metaclasses}.  Suppose that we have a set of $m$ metaclasses, $\mathcal{M}_m = \{\mathcal{C}_1,\dots,\mathcal{C}_m\}$, where each $\mathcal{C}_i$ ($i=1,\dots,m$) is a metaclass.  The metaclasses satisfy $\mathcal{C}_i \subset \mathcal{G}$ $(i=1,\dots,m)$, $\mathcal{C}_i \cap \mathcal{C}_j = \phi$ ($i \neq j$) and $ \cup_{i=1}^m \mathcal{C}_i = \{G_1,...,G_J\}$.  

On the basis of the metaclasses, we conduct two-stage LDA for given input $\bm{x}$.    First, we perform LDA to allocate input $\bm{x}$ to $\mathcal{C}_1,\cdots,\mathcal{C}_m$.  Suppose that $\bm{x}$ is allocated to $\mathcal{C}_i$.  If $\mathcal{C}_i$ consists of one class, i.e., $\mathcal{C}_i = \{G_j\}$ for some $j$, $\bm{x}$ is allocated to $G_j$.  If $\mathcal{C}_i$ consists of more than one class, expressed as $\mathcal{C}_i = \{G_{i_1},\dots,G_{i_j}\}$, we perform LDA on the metaclass $\mathcal{C}_i$ to allocate input $\bm{x}$ to $G_{i_1},\dots,G_{i_j}$.  The two-stage algorithm is summarized in Algorithm \ref{algorithm:Two step LDA}.   

The two-stage algorithm is implemented once a set of metaclasses $\mathcal{M} = \{\mathcal{C}_1,\dots,\mathcal{C}_m\}$ is determined.  To determine the metaclasses, we conduct hierarchical cluster analysis as follows.  The initial value of the metaclasses is $\mathcal{M}_J = \{\mathcal{C}_1,\dots,\mathcal{C}_J\}$ and $\mathcal{C}_j = \{G_j\}$ ($j=1,\dots,J$).   At each step, we combine two metaclasses $\mathcal{C}_j$ and $\mathcal{C}_k$ such that the CV value is minimized.   The details of the CV will be presented in Section \ref{sec:algorithm}; the definition is given by Eq. \eqref{eq:CVerror}.  The algorithm is summarized in Algorithm \ref{algorithm:Hierarchical LDA}.

\begin{algorithm}
%\floatname{algorithm}{a1.1}
	\caption{Two-stage LDA of a set of metaclasses $\mathcal{M}$}
	\label{algorithm:Two step LDA}
	\begin{algorithmic}[1]
		\STATE Input: $\bm{x}$
		\STATE Suppose that $\mathcal{M} = \{\mathcal{C}_1,...,\mathcal{C}_m\}$.  		
		\STATE Perform LDA on $\mathcal{M}$ to predict the class of $\bm{x}$. Suppose that $\bm{x}$ is allocated to $\mathcal{C}_j$.
		 \IF{$\mathcal{C}_j$ includes more than one class }
	%	 \IF{$\mathcal{C}_0$ consists of one class, say $G_k$s }
		 \STATE{Perform LDA on $\mathcal{C}_j$ and allocate $\bm{x}$ to a class of $\mathcal{C}_j$}
		 \ELSE
		 \STATE{Allocate $\bm{x}$ to $\mathcal{C}_j$}
		 \ENDIF 
	\end{algorithmic}
\end{algorithm}

\begin{algorithm}
%\floatname{algorithm}{a1.1}
	\caption{Hierarchical LDA.  The details of two-stage LDA are provided in Algorithm \ref{algorithm:Two step LDA}.}
	\label{algorithm:Hierarchical LDA}
	\begin{algorithmic}[1]
		\STATE Let $\mathcal{C}_j^0 = G_j$ $(j=1,\cdots,J)$. Define a set of metaclasses as $\mathcal{M}^{(0)} = \{\mathcal{C}_1^{(0)},...,\mathcal{C}_J^{(0)}\}$. 
	%	\STATE $\mathcal{G} = \{G_1,...,G_K\}$
	%	\STATE Conduct linear discriminant analysis
		\FOR{$t=0$ to $J-2$}
 		\FOR{all pairs of $\mathcal{C}_j $ and $ \mathcal{C}_k$ in $\mathcal{M}^{(t)} $}
		\STATE Create a new cluster $\mathcal{C}_{(j,k)} \leftarrow \mathcal{C}_j^{(t)} \cup \mathcal{C}_k^{(t)}$ 
		\STATE $\mathcal{A}_{(j,k)} \leftarrow \{\mathcal{M}^{(t)} ,\mathcal{C}_{(j,k)}\} \backslash \{\mathcal{C}_j^{(t)}, \mathcal{C}_k^{(t)}\}$ 
	%	\STATE $\mathcal{G} \leftarrow \mathcal{G} \backslash \{G_j, G_k\}$ 
		\STATE Perform CV in two-stage LDA based on a set of metaclasses $\mathcal{A}_{(j,k)}$, and obtain the CV value ${\rm CV}(j,k)$
 		\ENDFOR
 		\STATE Find an index $(j_0,k_0)$ whose CV value ${\rm CV}(j_0,k_0)$ is minimized.  
 		\STATE $\mathcal{C}_{(j_0,k_0)} \leftarrow \mathcal{C}_{j_0}^{(t)} \cup \mathcal{C}_{k_0}^{(t)}$, 
 		\STATE $\mathcal{M}^{(t+1)} \leftarrow \{\mathcal{M}^{(t)} ,\mathcal{C}\} \backslash \{\mathcal{C}_{j_0}^{(t)}, \mathcal{C}_{k_0}^{(t)}\}$ 
 		\STATE Denote the metaclass as $\mathcal{M}^{(t+1)} = \{\mathcal{C}_1^{(t+1)},\dots,\mathcal{C}_{K-t}^{(t+1)}\}$
 		\ENDFOR
		%\UNTIL $\mathcal{G} = \mathcal{C}_0$
	\end{algorithmic}
\end{algorithm}

\section{Efficient algorithm for CV of LDA}
\label{sec:algorithm}
CV is adopted to construct the metaclasses. Typically, $5$- or $10$-fold CV is used in classification problems. However, in practice, it would be unstable when $n_j$ is small, as shown in the case of the mouse consomic data.  Thus, we adopt leave-one-out CV.  

To implement CV, we compute the $D$ largest eigenvalues and eigenvectors of $(\bm{S}_{W,\delta}^{(-i)})^{-1/2}\bm{S}_B^{(-i)} (\bm{S}_{W,\delta}^{(-i)})^{-1/2}$, say $(\lambda_d^{(i)},\bm{s}_{d,\delta}^{(-i)})$ $(d=1,\dots,D;$ $i=1,\dots,n)$, where $\bm{S}_B^{(-i)} $ and $\bm{S}_{W,\delta}^{(-i)} $ are the between-class covariance and within-class covariance matrices constructed by $\{\bm{x}_1,\dots,\bm{x}_{i-1},\bm{x}_{i+1},\dots,\bm{x}_n\}$, respectively. Let $\bm{t}_{d,\delta}^{(-i)} = (\bm{S}_{W,\delta}^{(-i)})^{-1/2}\bm{s}_{d,\delta}^{(-i)}$ ($d=1,\dots,D$), and let  $\bm{T}_{\delta}^{(-i)} = (\bm{t}_{1,\delta}^{(-i)},\dots,\bm{t}_{D,\delta}^{(-i)})$.  The $p$-dimensional input of the test data, $\bm{x}_i$, is then transformed as $\bm{z}_i^{(-i)}=(\hat{\bm{T}}_{\delta}^{(-i)})^{\mathsf{T}} \bm{x}_i$. We estimate the class of $\bm{x}_i$, say $G^{(-i)}(\bm{x}_i)$, as follows:
\begin{equation}
	G^{(-i)}(\bm{x}_i)=\argmin_{j\in \{1,\ldots, J\}} \left(\bm{z}^{(-i)}-\bar{\bm{z}}_{j}^{(-i)}\right)^{\mathsf{T}} \left(\bm{z}^{(-i)}-\bar{\bm{z}}_{j}^{(-i)}\right), \label{eq:groupCV}
\end{equation}
where $ \bar{\bm{z}}_{j}^{(-i)}=\hat{\bm{T}}_{\delta}^{(-i)}\bar{\bm{x}}_j^{(-i)}$ with $\bar{\bm{x}}_j^{(-i)} = \frac{1}{n_j^{(-i)}}\sum_{k: y_k=j, k \neq i}^{}\bm{x}_k$.   
Here,
$$
n_j^{(-i)} = \left\{
\begin{array}{cc}
n_j - 1 & \quad (y_i = j)\\
n_j  & \quad (y_i \neq j)
\end{array}
.
\right.
$$
The CV value is calculated as 
\begin{equation}
\mbox{CV} =\frac{1}{n}\sum_{i=1}^{n} I\left[y_{i}\neq G^{(-i)}(\bm{x}_i)\right], \label{eq:CVerror}
\end{equation}
where $I(\cdot)$ is an indicator function.  

Our hierarchical LDA algorithm needs at least $O(J^2)$ operations of CV; hence, it becomes slow when $J$ is large.  Leave-one-out CV needs the eigenvalue and eigenvectors of $(\bm{S}_{W,\delta}^{(-i)})^{-1/2}\bm{S}_B^{(-i)} (\bm{S}_{W,\delta}^{(-i)})^{-1/2}$ for $i=1,\dots,n$, which involves a heavy computational load for large $p$.  Therefore, an efficient algorithm for computing CV is required.  In this section, we present an efficient algorithm by extending CV to two-class LDA described in \citet{Cawley:2003ha}.

\subsection{Least-squares formulation}
We consider the problem of ridge regression as follows:
\begin{eqnarray}
	\argmin_{\beta_0,\bm{\beta}}\left\{ \| \bm{y}  - \beta_0\bm{1}_n - \bm{X} \bm{\beta}\|^2 + \delta\| \bm{\beta}\|^2 \right\}\label{eq:ridge}
\end{eqnarray}
where $\bm{y} = (\bm{y}_{1},\dots,\bm{y}_{g})$ is a response vector, $\bm{1}_n$ is an $n$-dimensional vector whose elements are 1, $\beta_0$ is an intercept, $\bm{X} = (\bm{x}_1,\dots,\bm{x}_n)^{\mathsf{T}}$ is a design matrix, $\bm{\beta}$ is a regression coefficient vector, and $\delta \geq 0$ is a regularization parameter.  

The response vector $\bm{y}$ corresponds to the output of the classification problem.  For example, in the two-class classification problem (i.e., $J=2$), we may use $\bm{y}_1 = \frac{1}{n_1}\bm{1}_{n_1}$ and $\bm{y}_2 = -\frac{1}{n_2}\bm{1}_{n_2}$ \citep{Cawley:2003ha}.   With these response vectors, the estimate of the regression coefficient vector $\hat{\bm{\beta}}$ is parallel to the eigenvector of $\bm{S}_{W,\delta}^{-1}\bm{S}_B $. 

However, to the best of our knowledge, the construction of $\bm{y}$ for multiclass LDA has not been investigated thus far.  To show how the response vector $\bm{y}$ is constructed, first, we define $\lambda_d$ and $\bm{s}_d$ as the $d$th largest eigenvalue and the corresponding eigenvector of $\bm{S}_{W,\delta}^{-1/2}\bm{S}_B\bm{S}_{W,\delta}^{-1/2}$, respectively.  The response vectors $\bm{y}_{j}$ ($j=1,\dots,J$) are determined such that the ridge estimate of $\bm{\beta}$ is parallel to $\bm{S}_{W,\delta}^{-1/2}\bm{s}_d$.  

Because the responses and regression coefficients depend on $d$, we denote the responses and coefficients as $\bm{y}_d = (\bm{y}_{1d}^{\mathsf{T}},\dots,\bm{y}_{Jd}^{\mathsf{T}})$, $\beta_{0d}$, and $\bm{\beta}_d$.  The normal equation based on Eq. \eqref{eq:ridge} is then expressed as
\begin{eqnarray}
	\begin{pmatrix}
		%\displaystyle \sum_{j=1}^J\bm{X}_j^{\mathsf{T}}\bm{X}_j +  \delta \bm{I}_p  \ &  \ \displaystyle \sum_{j=1}^J\bm{X}_j^{\mathsf{T}}\bm{1}_{n_j} \\
		%\displaystyle \sum_{j=1}^J\bm{1}_{n_j}^{\mathsf{T}}\bm{X}_j& n \\
		n  &  \displaystyle \sum_{j=1}^J\bm{1}_{n_j}^{\mathsf{T}}\bm{X}_j \\
		\displaystyle \sum_{j=1}^J\bm{X}_j^{\mathsf{T}}\bm{1}_{n_j} \ & \ \displaystyle \sum_{j=1}^J\bm{X}_j^{\mathsf{T}}\bm{X}_j +  \delta \bm{I}_p   \\
	\end{pmatrix}
	\begin{pmatrix}
	\beta_{0d}\\
	\bm{\beta}_d 
	\end{pmatrix}
	=
	\begin{pmatrix}
	\bm{1}_{n_1}^{\mathsf{T}} &  \cdots & \bm{1}_{n_J}^{\mathsf{T}}
		\vspace{1mm}\\
	\bm{X}_1^{\mathsf{T}} &  \cdots & \bm{X}_J^{\mathsf{T}}\\
	\end{pmatrix}
	\begin{pmatrix}
\bm{y}_{1d}\\
\vdots\\
\bm{y}_{Jd}\\
	\end{pmatrix}.
	\label{eq:formula}
\end{eqnarray}
We define the response vectors $\bm{y}_{jd}$ $(j=1,\dots,J;\ d=1,\dots,D)$ as 
\begin{eqnarray}
	%\bm{y}_{jd} = \xi_{jd}\bm{1}_{n_j},\quad \xi_{jd} = \frac{1}{n\lambda_d} \sum_{k=1}^J n_k  (\bar{\bm{x}}_{j}-\bar{\bm{x}}_k)^{\mathsf{T}}\bm{S}_{W,\delta}^{-1/2} \bm{s}_d \label{eq:responses}
	\bm{y}_{jd} = \xi_{jd}\bm{1}_{n_j},\quad \xi_{jd} = \frac{1}{\lambda_d} (\bar{\bm{x}}_{j}-\bar{\bm{x}})^{\mathsf{T}}\bm{S}_{W,\delta}^{-1/2} \bm{s}_d. \label{eq:responses}
	%\quad (j=1,\dots,g;\ d=1,\dots,D)
\end{eqnarray}
Then, we get the following proposition: 
\begin{prop}
The right-hand side of Eq. \eqref{eq:formula} is calculated as
\begin{eqnarray*}
	 \sum_{j=1}^J \bm{1}_{n_j}^{\mathsf{T}}\bm{y}_{jd} &=& 0,\\
	 \sum_{j=1}^J \bm{X}_j^{\mathsf{T}}\bm{y}_{jd} &=& n\bm{S}_{W,\delta}^{1/2}\bm{s}_d.
\end{eqnarray*}
\label{prop:eigenreg}
\end{prop}
\begin{proof}
The proof is given in Appendix	\ref{app:prop:eigenreg}.
\end{proof}
With Proposition \ref{prop:eigenreg}, Eq. \eqref{eq:formula} is expressed as
\begin{eqnarray}
	 n\beta_{0d} + \sum_{j=1}^Jn_j\bar{\bm{x}}_j^{\mathsf{T}}\bm{\beta}_d  &=& 0,\label{eq:n2}\\
	\sum_{j=1}^Jn_j\bar{\bm{x}}_j\beta_{0d} + \left(\sum_{j=1}^J\bm{X}_j^{\mathsf{T}}\bm{X}_j + \delta\bm{I}_p\right) \bm{\beta}_d  &=&n\bm{S}_{W,\delta}^{1/2} \bm{s}_d. \label{eq:n1}
\end{eqnarray}
We have the following theorem:  
\begin{Theorem}\label{thm:betat}
Let $\hat{\bm{\beta}}_d$ be the solution of $\bm{\beta}_d$, i.e.,  $\hat{\bm{\beta}}_d$ is a regression coefficient that satisfies Eqs. \eqref{eq:n2} and \eqref{eq:n1}.   Let $\hat{\bm{t}}_{d,\delta}$ be the $d$th column vector of $\hat{\bm{T}}_{\delta}$.   Then, we obtain
\begin{equation*}
	\hat{\bm{\beta}}_d =\frac{1}{1+\lambda_d} \bm{S}_{W,\delta}^{-1/2} \bm{s}_d = \frac{1}{1+\lambda_d} \hat{\bm{t}}_{d,\delta}.
\end{equation*}
The estimate of regression coefficient vector $\hat{\bm{\beta}}_d$ is parallel to the transformation vector of multiclass LDA, $\hat{\bm{t}}_{d,\delta}$.
\end{Theorem}
\begin{proof}
	The proof is given in Appendix \ref{app:thm:betat}.
\end{proof}
The normal equation, Eq. \eqref{eq:formula}, is expressed as 
\begin{equation*}
	(\tilde{\bm{X}}^{\mathsf{T}}\tilde{\bm{X}} + \bm{\Delta})\bm{\alpha}_d = \tilde{\bm{X}}^{\mathsf{T}}\bm{y}_d,
\end{equation*}
where $\tilde{\bm{X}} = (\bm{1}_n,\bm{X})$, $\bm{\Delta} = {\rm diag}(0,\delta \bm{1}_D)$, $\bm{\alpha}_d = (\beta_{0d},\bm{\beta}_d^{\mathsf{T}})^{\mathsf{T}}$ ($d=1,\dots,D$), and $\bm{y}_{d} = (\bm{y}_{1d}^{\mathsf{T}},\dots,\bm{y}_{Jd}^{\mathsf{T}})^{\mathsf{T}}$.   The estimator of $\bm{\alpha}_d$ is then expressed as 
\begin{equation*}
	\hat{\bm{\alpha}}_d = (\tilde{\bm{X}}^{\mathsf{T}}\tilde{\bm{X}} + \bm{\Delta})^{-1}\tilde{\bm{X}}^{\mathsf{T}}\bm{y}_d. \label{eq:OLS}
\end{equation*}
Given input $\bm{x} \in \mathbb{R}^p$, denoting $\tilde{\bm{x}} = (1,\bm{x}^{\mathsf{T}})^{\mathsf{T}}$ and using Eq. \eqref{eq:n2}, we have the following relationship:
\begin{eqnarray}
	\tilde{\bm{x}}^{\mathsf{T}}\hat{\bm{\alpha}}_d &=&  \hat{\beta}_{0d}  + \bm{x}^{\mathsf{T}}\hat{\bm{\beta}}_{d} 
	%\\
	% &=& \bm{x}^{\mathsf{T}}\hat{\bm{\beta}}_d -\bar{\bm{x}}^{\mathsf{T}}\hat{\bm{\beta}}_d \\ 
	% &=& \frac{1}{1+\lambda_d}(\bm{x}-\bar{\bm{x}})^{\mathsf{T}}\hat{\bm{t}}_d, \label{eq:trans_regress}
	= \frac{1}{1+\lambda_d}(\bm{x}-\bar{\bm{x}})^{\mathsf{T}}\hat{\bm{t}}_{d,\delta}. \label{eq:trans_regress}
\end{eqnarray}
Therefore, the transformed vector $\hat{\bm{T}}_{\delta}\bm{x}$ is expressed as $(\bm{I}_D+\bm{\Lambda}_D)\hat{\bm{A}}\tilde{\bm{x}}$, where $\hat{\bm{A}} = (\hat{\bm{\alpha}}_1,\dots,\hat{\bm{\alpha}}_D)^{\mathsf{T}}$ and $\bm{\Lambda}_D = {\rm diag}(\lambda_1,\dots,\lambda_D)$.  
\subsection{Construction of an efficient algorithm}\label{sec:construction of algorithm}
We consider the problem of CV with the regression formulation in Eq. \eqref{eq:formula}.  We remove the $i$th observation from  $\tilde{\bm{X}}$ and $\bm{y}_d$, and construct an $(n-1) \times (p+1)$ matrix $\tilde{\bm{X}}^{(-i)} = (\tilde{\bm{x}}_1,\dots,\tilde{\bm{x}}_{i-1},\tilde{\bm{x}}_{i+1},\dots,\tilde{\bm{x}}_n)^{\mathsf{T}}$ and
\begin{equation*}
	\bm{y}_{jd}^{(-i)} := \xi_{jd}^{(-i)}\bm{1}_{n_j^{(-i)}},  \ \xi_{jd}^{(-i)} := \frac{1}{(n-1)\lambda_d^{(-i)}} \sum_{k=1}^J n_k^{(-i)}  (\bar{\bm{x}}_{j}^{(-i)}-\bar{\bm{x}}_k^{(-i)})^{\mathsf{T}} (\bm{S}_{W,\delta}^{(-i)})^{-1/2} \bm{s}_d^{(-i)}.
\end{equation*}
The ridge estimate based on $\tilde{\bm{X}}^{(-i)}$ and $\bm{y}_d^{(-i)}$ is expressed as 
\begin{equation*}
	\hat{\bm{\alpha}}_d^{(-i)} = \left\{(\tilde{\bm{X}}^{(-i)})^{\mathsf{T}}\tilde{\bm{X}}^{(-i)} + \bm{\Delta}\right\}^{-1}(\tilde{\bm{X}}^{(-i)})^{\mathsf{T}}\bm{y}_d^{(-i)}. \label{eq:alpha_-i_true}
\end{equation*}
Based on the regression coefficient vector $\hat{\bm{\alpha}}_d^{(-i)}$ and from Eq. \eqref{eq:trans_regress}, the transformed value for the $d$th dimension is
$$
	z_d - \bar{z}_d = (1+\lambda_d^{(-i)}) \left( \tilde{\bm{x}}_i- \tilde{\bar{\bm{x}}}_j^{(-i)}\right)^{\mathsf{T}}\hat{\bm{\alpha}}_d^{(-i)} ,
$$
From Eq. \eqref{eq:groupCV}, the class for the $i$th observation $(i=1,\dots,n)$ is  
\begin{equation}
	%{\rm argmin}_{j \in \{1,\dots,K\}} (z_i - \bar{z}_j)^2 =
	G^{(-i)}(\bm{x}_i) =  \argmin_{j \in \{1,\dots,J\}} \left\{ \sum_{d=1}^D (1+\lambda_d^{(-i)})^2 \left\{ (\tilde{\bm{x}}_i - \tilde{\bar{\bm{x}}}_j^{(-i)})^{\mathsf{T}}\hat{\bm{\alpha}}_d^{(-i)} \right\}^2 \right\}. \label{eq:allocCV_true}
\end{equation}
The problem with using Eq. \eqref{eq:allocCV_true} is that $\hat{\bm{\alpha}}_d^{(-i)}$ and $\lambda_d^{(-i)}$ require eigenvalues and eigenvectors of $(\bm{S}_{W,\delta}^{(-i)})^{-1/2}\bm{S}_B^{(-i)} (\bm{S}_{W,\delta}^{(-i)})^{-1/2}$, which involves $O(p^3)$ operations.   Therefore, direct calculation of both $(\tilde{\bm{x}}_i - \tilde{\bar{\bm{x}}}_j)^{\mathsf{T}}\hat{\bm{\alpha}}_d^{(-i)} $ and $\lambda_d^{(-i)}$ involves a heavy computational load when $p$ is large. To address these issues, we use $(n-1)$-dimensional vector $\bm{y}_d^{*(-i)} = (y_{1d},\dots,y_{(i-1)d},y_{(i+1)d},\dots,y_{nd})^{\mathsf{T}}$ instead of $\bm{y}_d^{(-i)}$.  Section \ref{sec:theoreticaljustification} presents the theoretical justification for using $\bm{y}_d^{*(-i)}$.  The ridge estimate based on $\tilde{\bm{X}}^{(-i)}$ and $\bm{y}_d^{*(-i)}$ is now defined as 
\begin{equation}
	\hat{\bm{\alpha}}_d^{*(-i)} = \left\{(\tilde{\bm{X}}^{(-i)})^{\mathsf{T}}\tilde{\bm{X}}^{(-i)} + \bm{\Delta}\right\}^{-1}(\tilde{\bm{X}}^{(-i)})^{\mathsf{T}}\bm{y}_d^{*(-i)}. \label{eq:alpha_-i}
\end{equation}
Furthermore, the approximation of $\lambda_d^{(-i)}$, say $\lambda_d^{*(-i)}$, is defined as 
\begin{eqnarray}
	 \lambda_d^{*(-i)} = \left\{\frac{1}{n-1}\sum_{k\neq i}^n (\tilde{\bm{x}}_k^{\mathsf{T}} \hat{\bm{\alpha}}_d^{*(-i)})^2  + \delta\hat{\bm{\beta}}_d^{*(-i)T}\hat{\bm{\beta}}_d^{*(-i)} \right\}^{-1} - 1,  \quad (i=1,\dots,n). \label{eq:lambdaast_def}
\end{eqnarray}	
The definition of $\lambda_d^{*(-i)}$ in Eq. \eqref{eq:lambdaast_def} is derived from the following relationship between $\bm{X}$ and $\lambda_d$:  
\begin{lemma}\label{lemma:x^a_eigenvalue}
The following equation holds:
\begin{eqnarray}
	\frac{1}{n}\sum_{i=1}^n (\tilde{\bm{x}}_i^{\mathsf{T}} \hat{\bm{\alpha}}_d)^2 + \delta \hat{\bm{\beta}}_d^{\mathsf{T}}\hat{\bm{\beta}}_d = \frac{1}{1+\lambda_d}.  \label{eq:lemma:x^a_eigenvalue0}
	%	\sum_{k\neq i}^n (\tilde{\bm{x}}_k \hat{\bm{\alpha}}_d^{(-i)})^2 = \frac{1}{1+\lambda_d^{(-i)}}  \label{eq:lemma:x^a_eigenvalue}
\end{eqnarray}	
\end{lemma}
\begin{proof}
The proof is presented in Appendix \ref{proof:lemma:x^a_eigenvalue}.	
\end{proof}
Based on the approximations given by Eqs. \eqref{eq:alpha_-i} and \eqref{eq:lambdaast_def}, the class allocation is now defined as 
\begin{equation}
	%{\rm argmin}_{j \in \{1,\dots,K\}} (z_i - \bar{z}_j)^2 =
	G^{*(-i)}(\bm{x}_i) = \argmin_{j \in \{1,\dots,J\}} \left\{ \sum_{d=1}^D (1+\lambda_d^{*(-i)})^2 \left\{ (\tilde{\bm{x}}_i - \tilde{\bar{\bm{x}}}_j)^{\mathsf{T}}\hat{\bm{\alpha}}_d^{*(-i)} \right\}^2 \right\}. \label{eq:allocCV} % \quad (i=1,\dots,n)
\end{equation}
%The theoretical justification for using $G^{*(-i)}(\bm{x}_i)$ instead of $G^{(-i)}(\bm{x}_i)$ will be presented in Section \ref{sec:theoreticaljustification}.

\subsubsection{Efficient computation of $G^{*(-i)}(\bm{x}_i)$.}
Direct computation of $G^{*(-i)}(\bm{x}_i)$ requires $O(p^3)$ operations ($i=1,\dots,n$) because we need to compute the inverse matrix $\left\{(\tilde{\bm{X}}^{(-i)})^{\mathsf{T}}\tilde{\bm{X}}^{(-i)} + \bm{\Delta}\right\}^{-1}$ to obtain $\hat{\bm{\alpha}}_d^{*(-i)} $.  However, the following theorem leads to efficient computation of $(\tilde{\bm{x}}_i - \tilde{\bar{\bm{x}}}_j)^{\mathsf{T}}\hat{\bm{\alpha}}_d^{*(-i)}$.  
\begin{Theorem}\label{thm:x^Talpha}
Given input vector $\tilde{\bm{x}} = (1,\bm{x}^{\mathsf{T}})^{\mathsf{T}}$, we have
	\begin{eqnarray}
	\hat{\bm{\alpha}}_d^{*(-i)}             &=& \hat{\bm{\alpha}}_d^* + \frac{\hat{y}_{id}-y_{id}}{1-h_{ii}}\bm{c}_i, \label{eq:ialpha}\\
	\tilde{\bm{x}}^{\mathsf{T}}\hat{\bm{\alpha}}_d^{*(-i)} &=& \sum_{k =1}^n y_{kd}h_k + \frac{(\hat{y}_{id}-y_{id})h_i}{1-h_{ii}} , \label{eq:xialpha}
\end{eqnarray}
where $\bm{c}_i=(\tilde{\bm{X}}^{\mathsf{T}}\tilde{\bm{X}} + \bm{\Delta})^{-1}\tilde{\bm{x}}_i$, $h_k = \tilde{\bm{x}} ^{\mathsf{T}} (\tilde{\bm{X}}^{\mathsf{T}}\tilde{\bm{X}} + \bm{\Delta})^{-1}\tilde{\bm{x}}_k$, $h_{ik} = \tilde{\bm{x}}_i^{\mathsf{T}}  (\tilde{\bm{X}}^{\mathsf{T}}\tilde{\bm{X}} + \bm{\Delta})^{-1}\tilde{\bm{x}}_k$, and $\hat{y}_{id} = \sum_{k=1}^nh_{ik}y_{kd}$. 
\end{Theorem}
\begin{proof}
The proof is presented in Appendix \ref{app:thm:x^Talpha}.
\end{proof}
Theorem \ref{thm:x^Talpha} implies that we do not need to compute $\left\{(\tilde{\bm{X}}^{(-i)})^{\mathsf{T}}\tilde{\bm{X}}^{(-i)} + \bm{\Delta}\right\}^{-1}$ $(i=1,\dots,n)$ to obtain $\tilde{\bm{x}}^{\mathsf{T}}\hat{\bm{\alpha}}_d^{*(-i)}$; we only need to calculate $(\tilde{\bm{X}}^{\mathsf{T}}\tilde{\bm{X}} + \bm{\Delta})^{-1}$.  Because $(\tilde{\bm{X}}^{\mathsf{T}}\tilde{\bm{X}} + \bm{\Delta})^{-1}$ does not depend on $i$, it can be computed before CV is conducted.

Using Theorem \ref{thm:x^Talpha}, we have
\begin{eqnarray}
	\tilde{\bm{x}}_i^{\mathsf{T}}\hat{\bm{\alpha}}_d^{*(-i)} &=& \frac{\hat{y}_{id} - y_{id}h_{ii}}{1 - h_{ii}}, \label{eq:x_ialpha}\\
	\tilde{\bar{\bm{x}}}_j^{(-i)T}\hat{\bm{\alpha}}_d^{*(-i)} &=& \frac{1}{n_j^{(-i)}}\sum_{k:y_k=j,k\neq i}\hat{y}_{kd} + \frac{\hat{y}_{id}-y_{id}}{1-h_{ii}} \left(\frac{1}{n_j^{(-i)}}\sum_{k:y_k=j,k\neq i}h_{ki}\right),  \label{eq:x_jbaralpha}\\
%\end{eqnarray}
%where 
%$$
%n_j^{(-i)} = \left\{
%\begin{array}{cc}
%n_j - 1 & \quad (y_i = j)\\
%n_j  & \quad (y_i \neq j)
%\end{array}
%.
%\right.
%$$
%\begin{eqnarray}
	\sum_{k \neq i}^n(\tilde{\bm{x}}_k^{\mathsf{T}}\hat{\bm{\alpha}}_d^{*(-i)})^2 &=& \sum_{k \neq i}^n \left( \sum_{r =1}^n y_{rd}h_{kr} + \frac{\hat{y}_{id}-y_{id}}{1-h_{ii}} h_{ki} \right)^2 \nonumber \\
	&=& \sum_{k \neq i}^n \hat{y}_{kd}^2 +2a_i \sum_{k \neq i}^n \hat{y}_{kd} h_{ki} +a_i^2 \sum_{k \neq i}^n h_{ki}^2  ,\label{eq:x^Talpha_-i_eigenvalue}
%	&=& \sum_{k=1}^n \hat{y}_{kd} - \hat{y}_{id} + \frac{(\hat{y}_{id}-y_{id})}{1-h_{ii}}  \left(\sum_{k =1}^n h_{ki} - h_{ii}\right),\label{eq:x^Talpha_-i_eigenvalue}
\end{eqnarray}
where  $\bm{c}_i=(\tilde{\bm{X}}^{\mathsf{T}}\tilde{\bm{X}} + \bm{\Delta})^{-1}\tilde{\bm{x}}_i$ and $a_i=\frac{\hat{y}_{id}-y_{id}}{1-h_{ii}}$.  These values are required to compute $G^{*(-i)}(\bm{x}_i)$. Using the above-mentioned formulae is far more efficient than direct computation of $\left\{(\tilde{\bm{X}}^{(-i)})^{\mathsf{T}}\tilde{\bm{X}}^{(-i)} + \bm{\Delta}\right\}^{-1}$ because it requires only $O(n)$ operations once $\hat{\bm{y}}$ and $\bm{H}$ are computed.

Algorithm \ref{algorithm:CV} summarizes our efficient algorithm for CV of multiclass LDA.  
\begin{algorithm}
%\floatname{algorithm}{a1.1}
	\caption{Efficient algorithm for CV of LDA}
	\label{algorithm:CV}
	\begin{algorithmic}[1]
		\STATE Calculate $\bm{S}_{W,\delta}^{-1/2}\bm{S}_B\bm{S}_{W,\delta}^{-1/2}$, and then obtain the $D$ largest eigenvalues and the corresponding eigenvectors, say $\lambda_d$ and $\bm{s}_d$ ($d=1,\dots,D$).  
		\STATE Calculate $\bm{H}  = (h_{ij}) = \bm{X}(\bm{X}^{\mathsf{T}}\bm{X} + \Delta)^{-1}\bm{X}^{\mathsf{T}}$ .  
%		\STATE Calculate $\gamma_{jk} = \frac{1}{n\lambda_d}	n_{j}n_k (\bm{a}_{jk}^{\mathsf{T}} \bm{s}_d )$, where $\bm{a}_{jk} = \bm{S}_{W,\delta}^{-1/2} (\bar{\bm{x}}_{j}-\bar{\bm{x}}_k)$ \quad $(k,j=1,\dots,g)$
%		\STATE Calculate $\bm{y}_{jd} = \frac{\xi_j}{n_j}\bm{1}_{n_j}$, where $\xi_j = \sum_{k\neq j}^{g}(\gamma_{jk} - \gamma_{kj})$  $(j=1,\dots,p$).
		\FOR{$d=1$ to $D$}
		\STATE Calculate the response vector $\bm{y}_d =(\bm{y}_{1d}^{\mathsf{T}},\dots,\bm{y}_{Jd}^{\mathsf{T}})^{\mathsf{T}}$ as follows:
		$$	\bm{y}_{jd} = \xi_{jd}\bm{1}_{n_j},\quad \xi_{jd} = \frac{1}{\lambda_d} (\bar{\bm{x}}_{j}-\bar{\bm{x}})^{\mathsf{T}}\bm{S}_{W,\delta}^{-1/2} \bm{s}_d \quad (j=1,\dots,J)
		$$
		\STATE Calculate 
		%$\hat{\bm{\alpha}}_d = (\tilde{\bm{X}}^{\mathsf{T}}\tilde{\bm{X}} + \bm{\Delta})^{-1}\tilde{\bm{X}}^{\mathsf{T}}\bm{y}_d$, and 
		$\hat{\bm{y}}_d = \bm{H}\bm{y}_d$.
		\ENDFOR
		\FOR{$i=1$ to $n$}
		\FOR{$d=1$ to $D$}
		\STATE Compute $\tilde{\bm{x}}_i^{\mathsf{T}}\hat{\bm{\alpha}}_d^{*(-i)}$, $\tilde{\bar{\bm{x}}}_j^{\mathsf{T}}\hat{\bm{\alpha}}_d^{*(-i)}$, $\sum_{k \neq i}^n(\tilde{\bm{x}}_k^{\mathsf{T}}\hat{\bm{\alpha}}_d^{*(-i)})^2$, and $\hat{\bm{\beta}}_d^{*(-i)T}\hat{\bm{\beta}}_d^{*(-i)} $ by Eqs. \eqref{eq:ialpha} and \eqref{eq:x_ialpha}--\eqref{eq:x^Talpha_-i_eigenvalue}, and obtain $(1+\lambda_d^{*(-i)})^2 \left\{ (\tilde{\bm{x}}_i - \tilde{\bar{\bm{x}}}_j)^{\mathsf{T}}\hat{\bm{\alpha}}_d^{*(-i)} \right\}^2 $.
		\ENDFOR
\STATE %Determine the class for $\bm{x}_i$ by
Allocate $\bm{x}_i$ to class $G^{*(-i)}(\bm{x}_i)$ using Eq. \eqref{eq:allocCV}. 
\ENDFOR
		\STATE Calculate the CV value as follows:
		\begin{equation*}
\mbox{CV} =\frac{1}{n}\sum_{i=1}^{n} I\left[y_{i}\neq G^{(-i)}(\bm{x}_i)\right].
\end{equation*}

		%\UNTIL $\mathcal{G} = \mathcal{C}_0$
	\end{algorithmic}
\end{algorithm}

\subsection{Theoretical justification}\label{sec:theoreticaljustification}
Algorithm \ref{algorithm:CV} is not exactly equivalent to CV in LDA because we use $\bm{y}_{jd}^{*(-i)}$ and $\lambda_{d}^{*(-i)}$ instead of $\bm{y}_{jd}^{(-i)}$ and $\lambda_{d}^{(-i)}$, respectively.  Thus, the class allocations $G^{*(-i)}(\bm{x}_i)$ and $G^{(-i)}(\bm{x}_i)$ can be different.  However, $G^{*(-i)}(\bm{x}_i)$ and $G^{(-i)}(\bm{x}_i)$ are shown to be asymptotically equivalent.  We consider asymptotics where the number of observations for each class is sufficiently large while the number of classes is fixed.
\begin{Assumption}\label{assumption}
Consider the case where $n_j/n \rightarrow c_j \in (0,\infty)$ as $n \rightarrow \infty$ ($j=1,\dots,J$). 	We assume that $\bar{\bm{x}}_j {\xrightarrow{\rm a.s.}} \bm{\mu}_j$ and $\sum_{i:y_i=j} \bm{x}_i\bm{x}_i^{\mathsf{T}}/n_j {\xrightarrow{\rm a.s.}} \bm{A}_{j}$ ($j=1,\dots,J$).  
\end{Assumption}
Under Assumption \ref{assumption}, we obtain the following proposition.
\begin{prop}\label{prop:consistency}
	Under Assumption \ref{assumption}, we have
\begin{eqnarray}
&\xi_{jd} \as  \xi_{0,jd},\label{eq:propconsistency1}\\
	& (1+\lambda_d^{(-i)})^2\left\{(\tilde{\bm{x}}_i - \tilde{\bar{\bm{x}}}_j^{(-i)})^{\mathsf{T}}\hat{\bm{\alpha}}_d^{(-i)}\right\}^2 
	 - (1+\lambda_d^{*(-i)})^2\left\{(\tilde{\bm{x}}_i - \tilde{\bar{\bm{x}}}_j^{(-i)})^{\mathsf{T}}\hat{\bm{\alpha}}_d^{*(-i)} \right\}^2 \as 0.\label{eq:propconsistency2}
\end{eqnarray}
Here, $ \xi_{0,jd}$ is a constant value.
\end{prop}
\begin{proof}
	The proof is presented in Appendix \ref{sec:proofconsistency}.
\end{proof}
Since $\bm{y}_{jd} = \xi_{jd}\bm{1}_{n_j}$ as in Eq. \eqref{eq:responses}, Eq. \eqref{eq:propconsistency1} implies that $\bm{y}_{jd}^{*(-i)}$ and $\bm{y}_{jd}^{(-i)}$ are asymptotically equivalent.  Furthermore, Eq. \eqref{eq:propconsistency2}} suggests that $\bm{y}_{jd}^{*(-i)}$ and $\bm{y}_{jd}^{(-i)}$ asymptotically provide the same class allocation for an observation $\bm{x}_i$ because the class allocation rule is provided by Eq. \eqref{eq:allocCV}.  Therefore, it would be reasonable to perform CV using $\bm{y}_{jd}^{*(-i)}$ instead of $\bm{y}_{jd}^{(-i)}$.  

In Section \ref{sec:comparisonCV}, we investigate the numerical performance of CV. Our results show that the approximation works well when $n$ is large.   

\section{Numerical experiments with artificial data}
\label{sec:numerical}
\subsection{Data generation}
\label{sec:simsetup}
In the numerical experiment, the label of the $i$th observation, $y_i$, is generated according to a multinomial distribution with probability $P(y_{i}=j)=1/J$ $(i=1,\dots,n; \ j=1,\ldots, J)$. We then define the mean vector of each cluster, say $\bm{\mu}_{j}$ ($j=1,\dots,J$).  Given label $y_i$, the $i$th predictor vector is generated from 
\begin{align*}
	\bm{x}_i \sim N_{p}(\bm{\mu}_{y_i},\bm{I}_p), \ (i = 1, \ldots ,n).
\end{align*}
Here, two simulation models are considered as follows.  

\begin{figure}[!t]
    \subfigure[Data points generated by Model 1]{\includegraphics[width=8cm,bb=0 0 360 360]{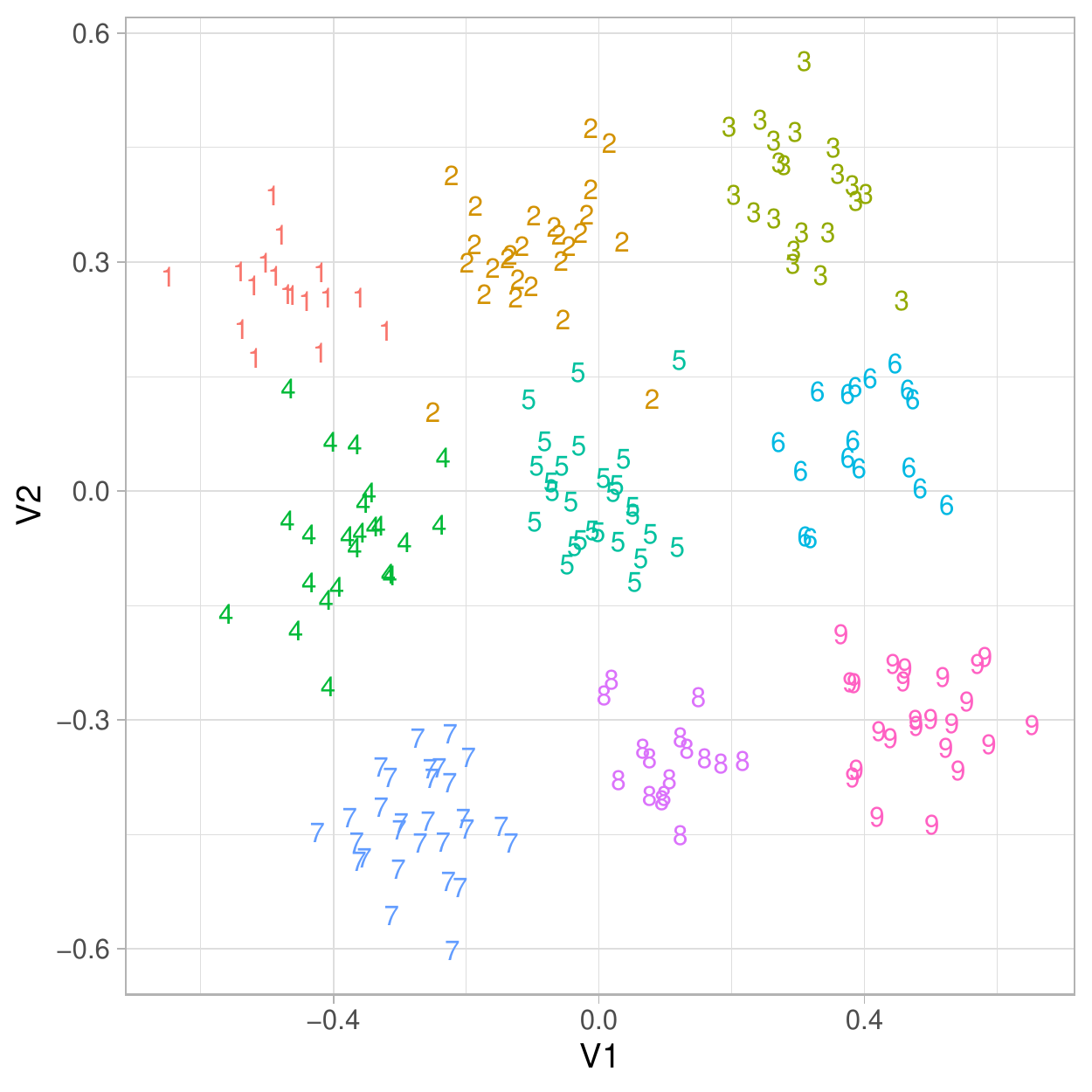}\label{fig:giji_rand}}
    \subfigure[2D projected data points generated by Model 2]{\includegraphics[width=8cm,bb=0 0 360 360]{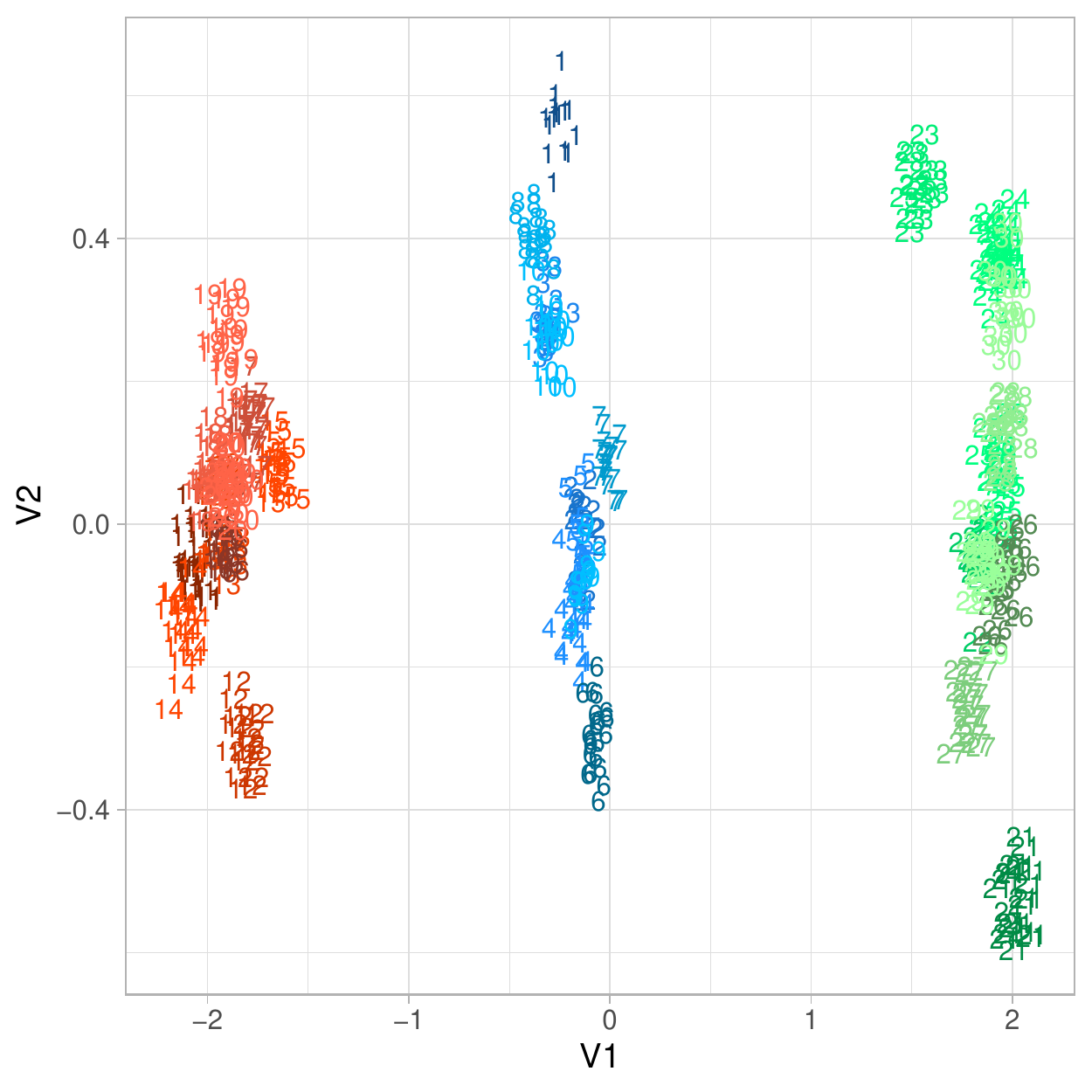}\label{fig:pl_giji}}
    \caption{2D plot of artificial dataset generated from (a) Model 1 (a) and (b) Model 2.}
    \label{fig:2Dplot_artificial}
\end{figure}

\begin{description}
	\item[Model 1] We set $J = 9$, $p=2$, and $n=200$.  The mean vectors $\bm{\mu}_j$ $(j=1,\dots,9)$ are defined as 
$$\bm{\mu}_j = 5
	\begin{pmatrix}
	\displaystyle
	\left\lfloor \frac{j-1}{3} \right\rfloor - 1, & 	
	\displaystyle
	j - 2 -3\left\lfloor \frac{j-1}{3} \right\rfloor
	\end{pmatrix}^{\mathsf{T}}.
$$
%\begin{align*}
%	\begin{array}{ccc}
%	\bm{\mu}_1 = (-5 ,-5)^{\mathsf{T}}, &\bm{\mu}_2 = (-5, 0)^{\mathsf{T}}, &\bm{\mu}_3 = (-5, 5)^{\mathsf{T}}, \\
%	\bm{\mu}_4 = (0 ,-5)^{\mathsf{T}}, &\bm{\mu}_5 = (0, 0)^{\mathsf{T}}, &\bm{\mu}_6 = (0, 5)^{\mathsf{T}}, \\
%	\bm{\mu}_7 = (5 ,-5)^{\mathsf{T}}, &\bm{\mu}_8 = (5, 0)^{\mathsf{T}}, &\bm{\mu}_9 = (5, 5)^{\mathsf{T}}
%	\end{array}
%\end{align*}
Figure \ref{fig:giji_rand} shows the data points generated from Model 1. Each class is well separated from the others; thus, ordinary LDA is expected to perform well when $D=2$.  We investigate whether our method, HLDA, performs well even when the clusters are not necessarily required for classification.  
	\item[Model 2] Let $\bm{c}_j$ $(j = 1, \dots, J)$ be
\begin{eqnarray*}
\bm{c}_j = \left\{
\begin{array}{ll}
	\bm{1}_{p} & (j=1,\dots,10),\\
	10\cdot \bm{1}_{p} & (j=11,\dots,20),\\
	-10\cdot \bm{1}_{p} & (j=21,\dots,30),
\end{array}
\right.
\end{eqnarray*}
and the mean vectors $\bm{\mu}_j (j=1, \ldots , J)$ are generated from
\begin{align*}
	\bm{\mu}_j \sim N_{p}(\bm{c}_j ,10\cdot\bm{I}_p).
\end{align*}
Therefore, the mean vector belongs to one of the three clusters of the centers: $\bm{1}_{p}$, $10\cdot \bm{1}_{p}$, and $-10\cdot \bm{1}_{p}$.  Here, we set $J = 30$, $p=20$, and $n=600$, except for Section \ref{sec:comparisonCV}, in which the performance of our fast CV is investigated.

The 2D projection of the data points is shown in Figure \ref{fig:pl_giji}.   In this case, ordinary LDA may not perform well for $D=2$ because of the difficulty in classification within a cluster.  We expect HLDA to perform better than LDA.  
	\end{description}
We employ {\tt R 4.0.2} to implement our algorithm.  The {\tt RCpp} package is used for fast computation. The OpenBLAS library is used for matrix computation.  
For implementation, we use Amazon Web Services (AWS) with Intel Xeon Platinum 8175 processors (3.1 GHz), $32$ vCPUs, $128$ GB memory, and CentOS $7$.  Throughout the experiments, the ridge parameter is $\delta =10^{-5}$.

%\begin{align*}
%	\begin{array}{ccc}
%	\bm{\mu}_1 = (-5 ,-5)^{\mathsf{T}}, &\bm{\mu}_2 = (-5, 0)^{\mathsf{T}}, &\bm{\mu}_3 = (-5, 5)^{\mathsf{T}}, \\
%	\bm{\mu}_4 = (0 ,-5)^{\mathsf{T}}, &\bm{\mu}_5 = (0, 0)^{\mathsf{T}}, &\bm{\mu}_6 = (0, 5)^{\mathsf{T}}, \\
%	\bm{\mu}_7 = (5 ,-5)^{\mathsf{T}}, &\bm{\mu}_8 = (5, 0)^{\mathsf{T}}, &\bm{\mu}_9 = (5, 5)^{\mathsf{T}}
%	\end{array}
%\end{align*}
%Figure \ref{giji_rand} shows the two-dimenisonal plots generated from this model. 

\subsection{Illustration of our method}

With our HLDA algorithm, the number of clusters is determined by the number of steps of the hierarchical clustering algorithm, $t$; the number of clusters is $J-t$.  Note that HLDA is equivalent to LDA when $t=0$ and $t=J-1$.  

We first apply the HLDA algorithm to the dataset shown in Figure \ref{fig:2Dplot_artificial}.   Figure \ref{fig:error_sim_dim} shows the CV error as a function of $t$ ($t=0, \dots, J-1$).  The dimension of the projected data, $D$, is $D=1,2$ for Model 1 and $D=2$ for Model 2.    We note that the projected data points are equivalent to the original data points for Model 1 when $D=2$.  
\begin{figure}[!t]
\centering
    \subfigure[Model 1, $D=1$]{\includegraphics[width=5cm,bb=0 0 504 504]{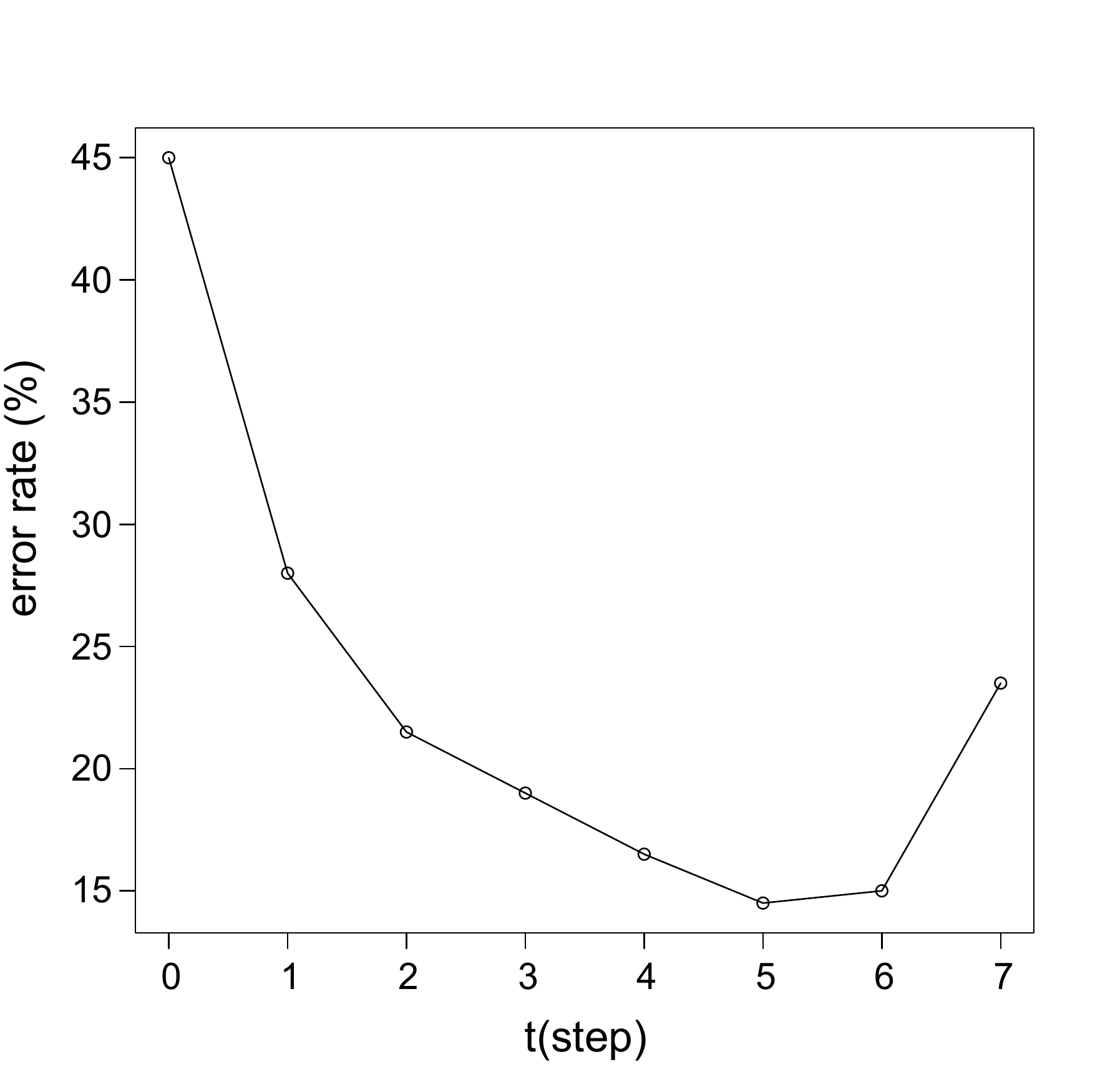}\label{fig:hiera_errorbar_giji1}}
    \subfigure[Model 1, $D=2$]{\includegraphics[width=5cm,bb=0 0 504 504]{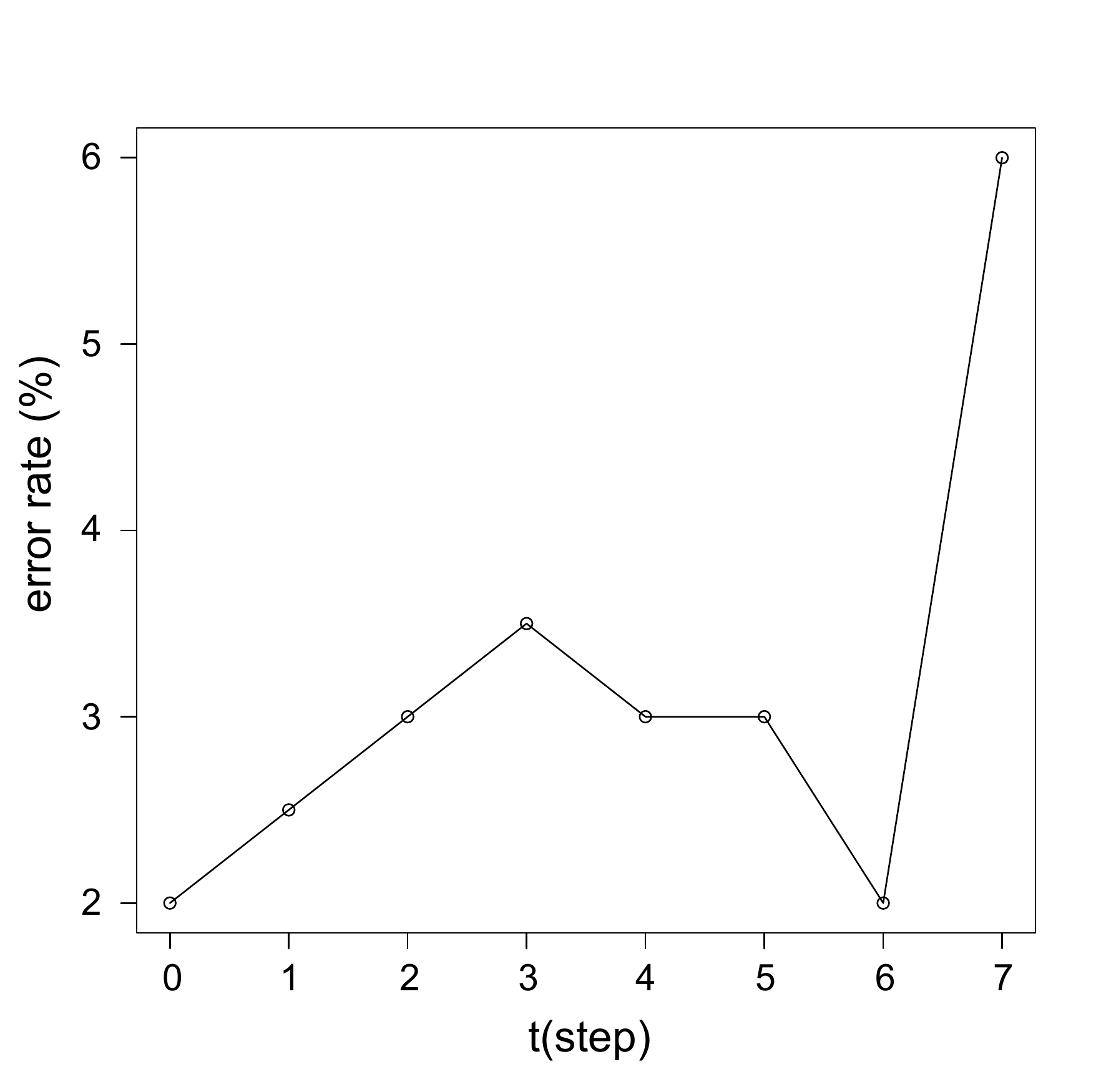}\label{fig:hiera_errorbar_giji2}}
    \subfigure[Model 2, $D=2$]{\includegraphics[width=5cm,bb=0 0 504 504]{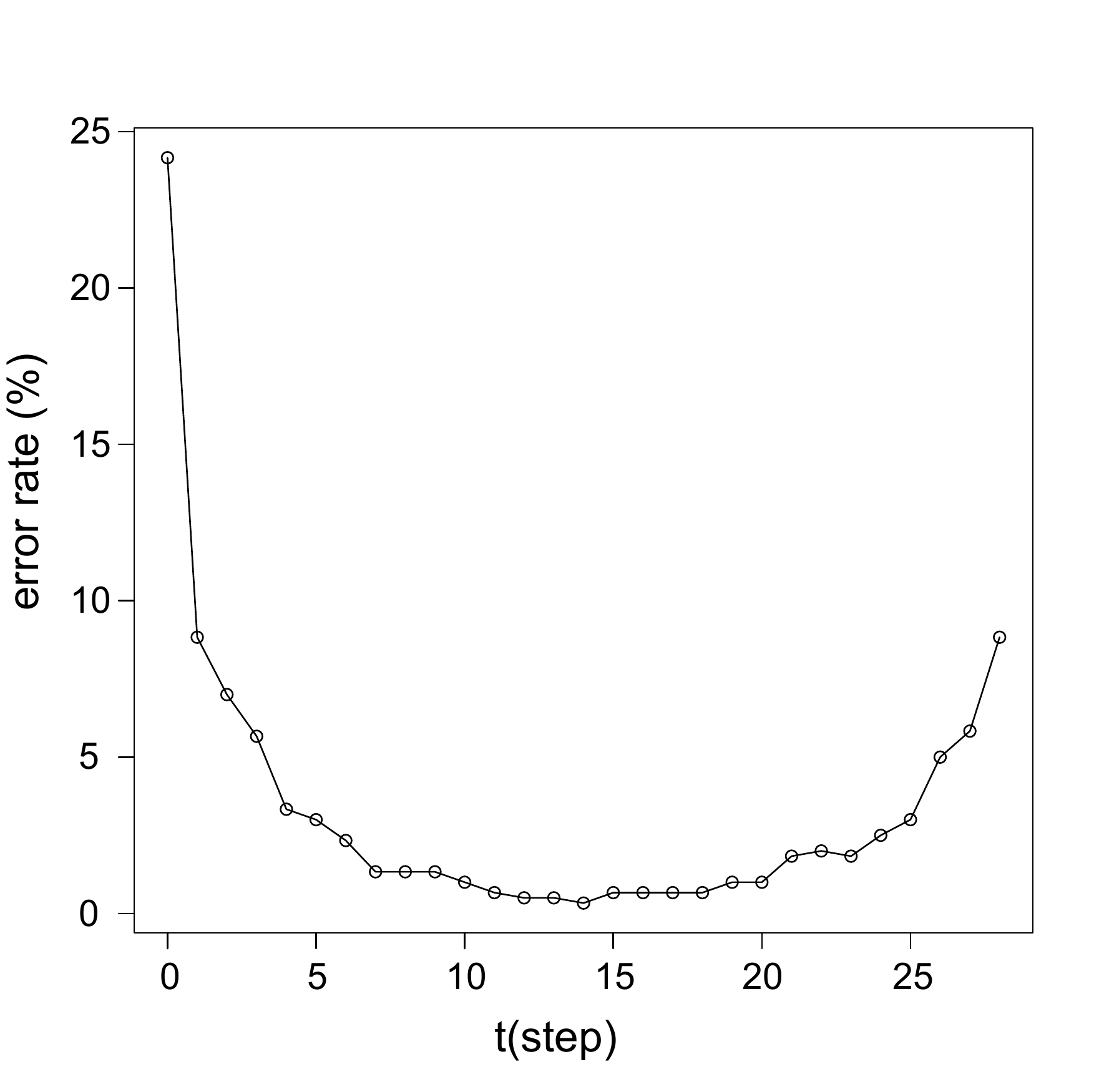}\label{fig:hiera_errorbar_giji3}}
    \caption{CV error as a function of $t$  ($t=0, \dots, J-2$).}
    \label{fig:error_sim_dim}
\end{figure}

For Model 1, ordinary LDA ($t=0$) with $D=1$ does not perform well because any one-dimensional projection causes a high CV error rate. However, HLDA results in high accuracy when $t \geq 1$; the clustering considerably improves the error rate. Meanwhile, when $D=2$, LDA performs well because the linear classification boundaries are easily constructed without clustering; therefore, the cluster analysis does not always decrease the CV error rate.  Indeed, CV selects $t=0$, which is identical to the ordinary LDA classification rule.  Thus, HLDA selects an appropriate classification rule even if clustering is not required.  For Model 2, our HLDA algorithm significantly decreases the CV error rate when $t \geq 1$, which suggests that clustering improves the prediction accuracy considerably.   The number of clusters should be relatively large to ensure a sufficiently low CV error rate.

\subsection{Prediction accuracy}
We compare the prediction accuracies of three methods: LDA, HLDA, and Ward's hierarchical clustering. The training set with $n$ observations is generated, and the classification rules for the three above-mentioned methods are then constructed. The test data with $n$ observations are generated from the same distribution as the training set, and the error rate is computed. This procedure is repeated 100 times.  

Figure \ref{fig:error_sim_dim100} shows the error rates for 100 simulated datasets as a function of $D$, the dimension of the projected space.  The error bar represents one standardization error.  For each $D$, the number of clusters of HLDA is determined such that the CV value is minimized.  The number of clusters of Ward's method is the same as that selected by HLDA.
\begin{figure}[!t]
\centering
    \subfigure[Model 1]{\includegraphics[height=7cm,bb=0 0 354 473]{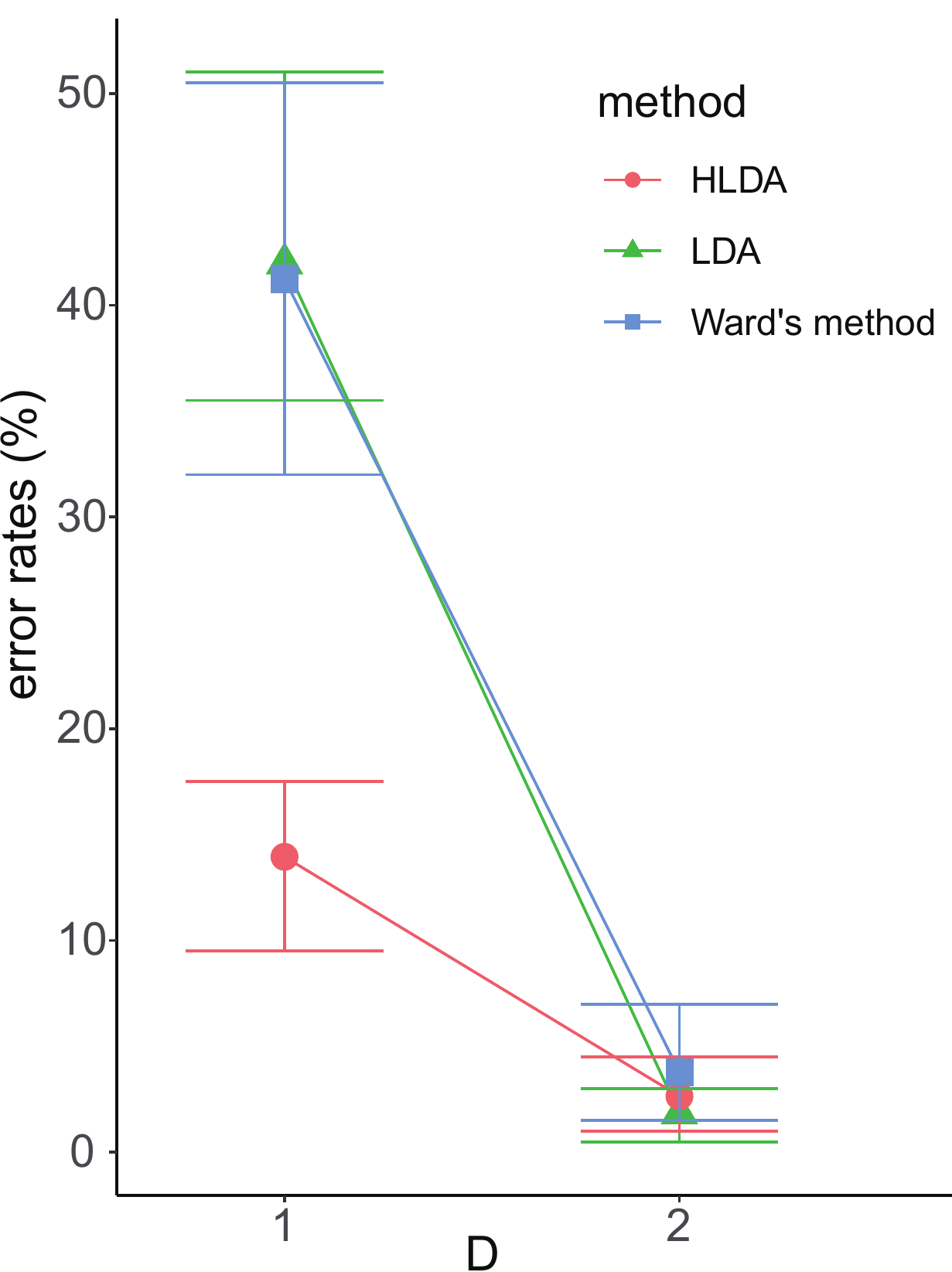}\label{fig:hiera_errorbar_giji_D1}}
    \hspace{1.2cm}
    \subfigure[Model 2]{\includegraphics[height=7cm,bb=0 0 609 474]{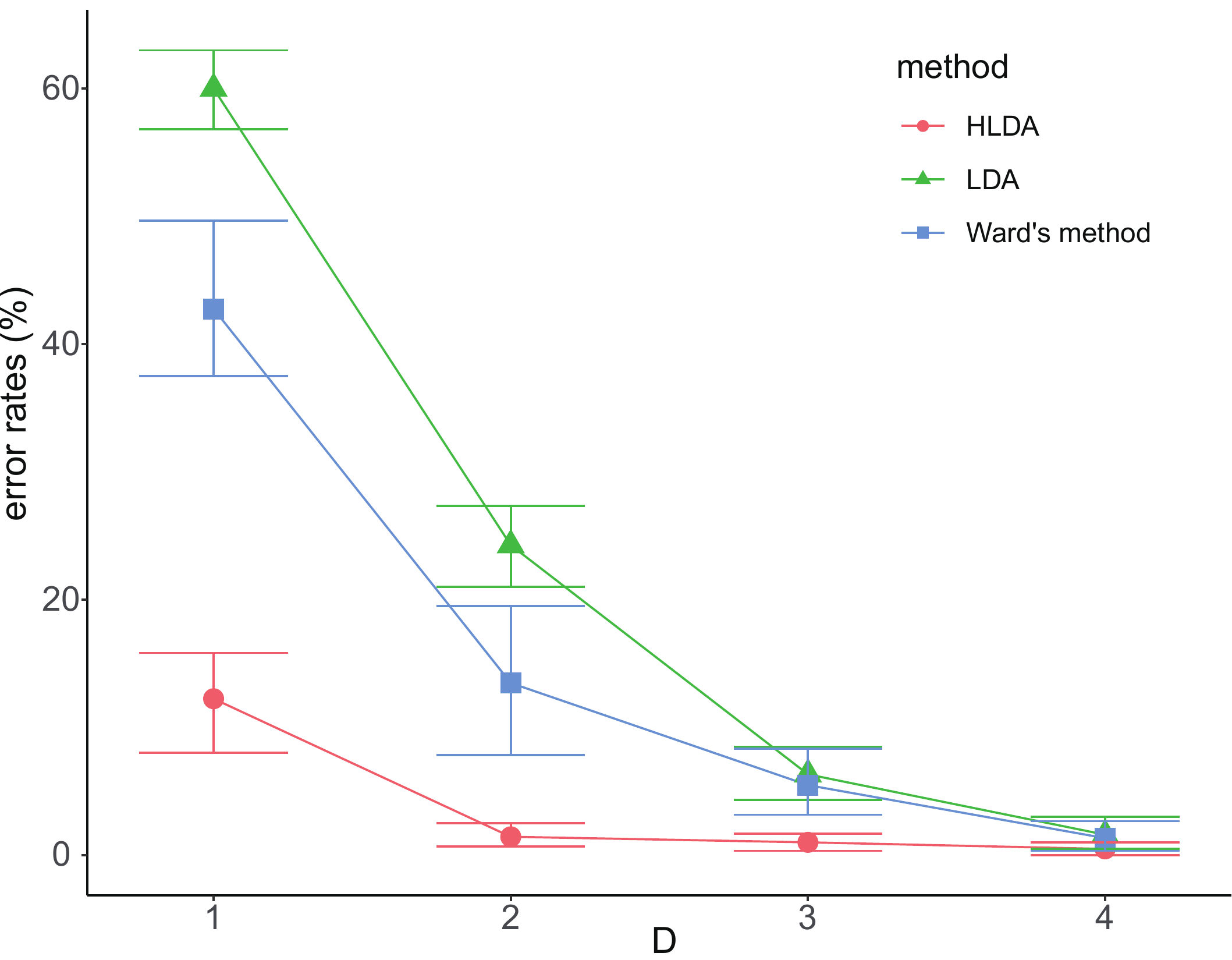}\label{fig:hiera_errorbar_giji_D2}}
    \caption{Comparison of error rates as a function of $D$, the dimension of the projected space.  The error bar represents one standardization error.}
    \label{fig:error_sim_dim100}
\end{figure}

For Model 1, HLDA performs much better than existing methods when $D=1$.  Therefore, our clustering method significantly improves the accuracy.  When $D=2$, LDA is expected to perform well, as shown in Figure \ref{fig:giji_rand}.  Indeed, LDA achieves good accuracy when $D=2$.  Nevertheless, HLDA and LDA provide nearly identical performances; HLDA may perform well even if clustering is not required.  Ward's method performs well but is slightly worse than LDA and HLDA.  

For Model 2, HLDA performs much better than existing methods, especially when $D \leq 3$.   When $D=4$, LDA performs well but is outperformed by HLDA.  We observe that the error rates become nearly 0\% within only a few steps of the HLDA algorithm when $D = 4$.

Figure \ref{fig:sim_HLDA_Ward_comparison} depicts the clustering results of the dataset shown in Figure \ref{fig:pl_giji}.   The number of clusters is 3 (i.e., $t=27$) because we assume three clusters of mean vectors in Model 2.  We remark that the 2D plots are depicted by 1 out of 100 datasets, and we observe similar tendencies for most of the datasets.  
\begin{figure}[!t]
\centering
    \subfigure[HLDA]{\includegraphics[width=7cm,bb=0 0 360 360]{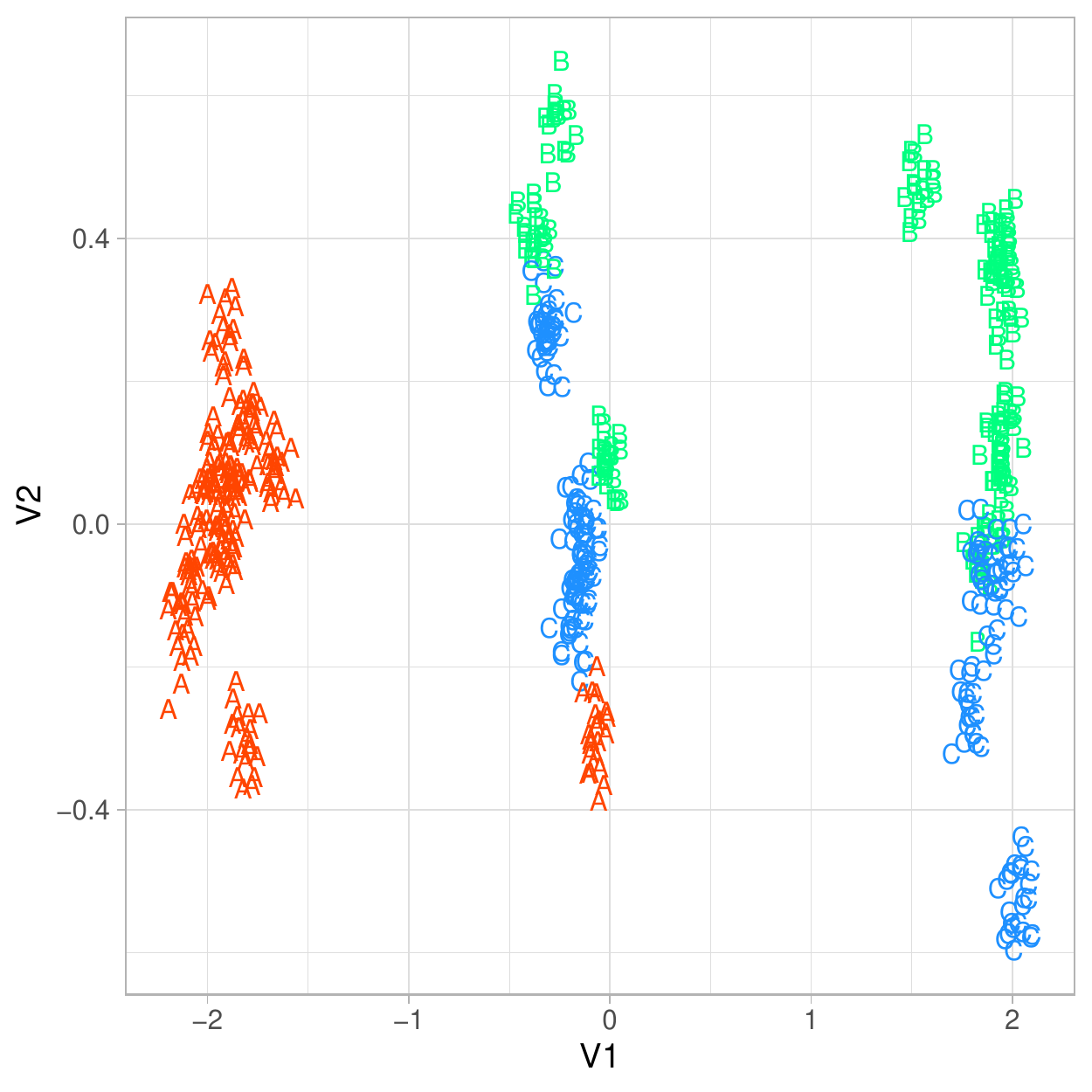}\label{fig:p_giji_clu2}}
       \hspace{1cm}
 \subfigure[Ward's method]{\includegraphics[width=7cm,bb=0 0 360 360]{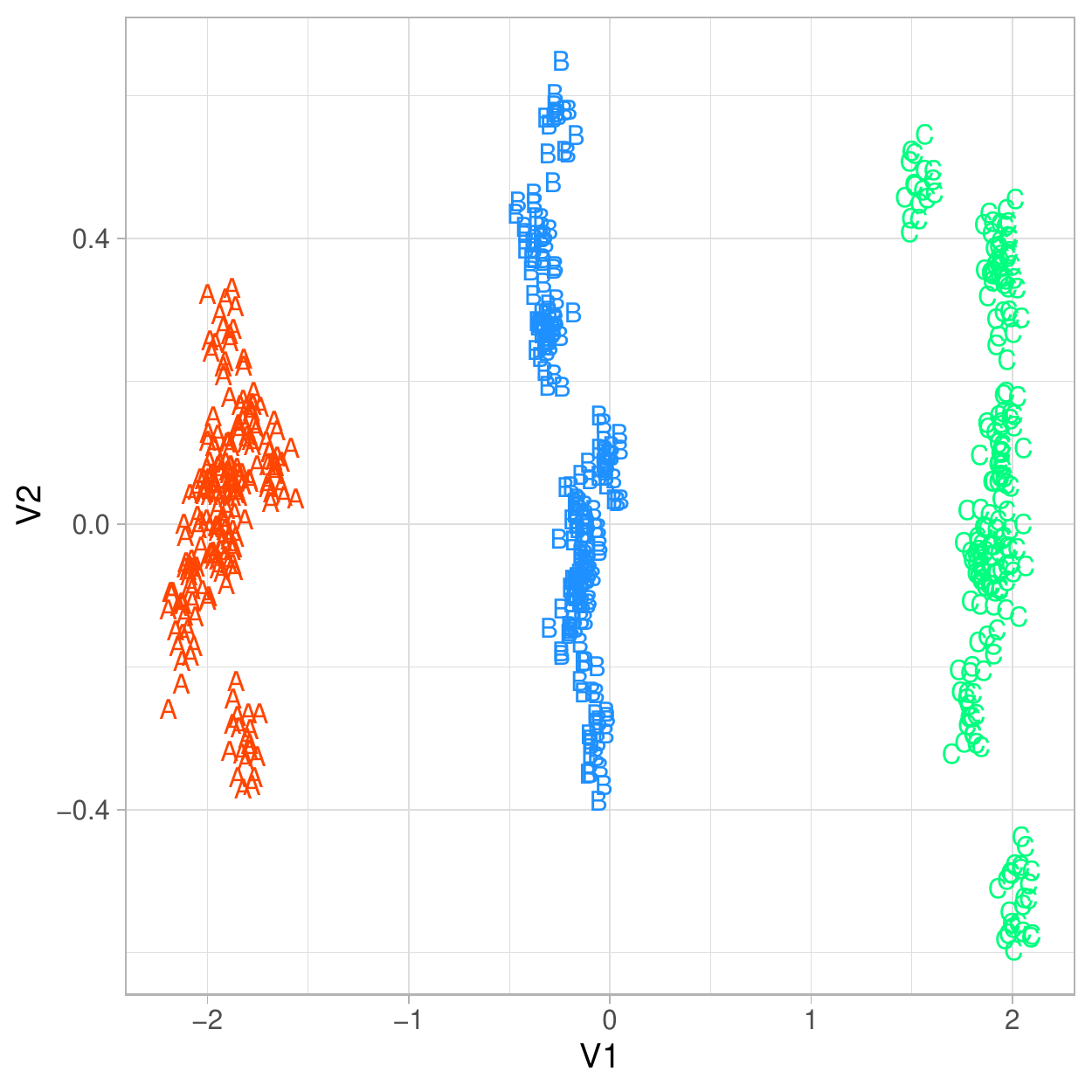}\label{fig:p_giji_clu_kmeans}}
    \caption{Comparison of clustering results obtained by (a) HLDA and (b) Ward's method for Model 2.  The 2D projection is performed by LDA.  The number of clusters is 3 (i.e., $t=27$), and the clusters are labeled as ``A", ``B", and ``C."}
    \label{fig:sim_HLDA_Ward_comparison}
\end{figure}
The results show that HLDA produces different clusters compared to Ward's method.  In particular, HLDA often creates a cluster whose classes are separated, whereas Ward's method cannot produce such a cluster.  The error rates of HLDA and Ward's method are 4.83\% and 7.17\%, respectively; thus, HLDA achieves higher accuracy than Ward's method.  
\subsection{Comparison of CV computation}\label{sec:comparisonCV}
We discuss the numerical investigation of our fast CV computation, described in Section \ref{sec:algorithm}.  In particular, we focus on the comparison of the computational timings and the accuracy of fast CV.  
\subsubsection{Comparison of computational timings}
For the timing comparison,  the datasets are generated from Model 2 described in Section \ref{sec:simsetup} with various dimensions: $p=20, 50 ,100, 200, 500, 1000$. The number of observations, $n$, is $n=6000$. We employ the fast and ordinary CV algorithms and compare their computational timings over 10 runs.  

Figure \ref{fig:timing_err_bar} shows the computational timings of the fast and ordinary CV algorithms.  The error bar represents the maximum and minimum values over 10 runs.  
\begin{figure}[!t]
	\begin{center}
    \subfigure[Computational timing as a function of $p$.]{\includegraphics[height=5cm,bb=0 0 465 338]{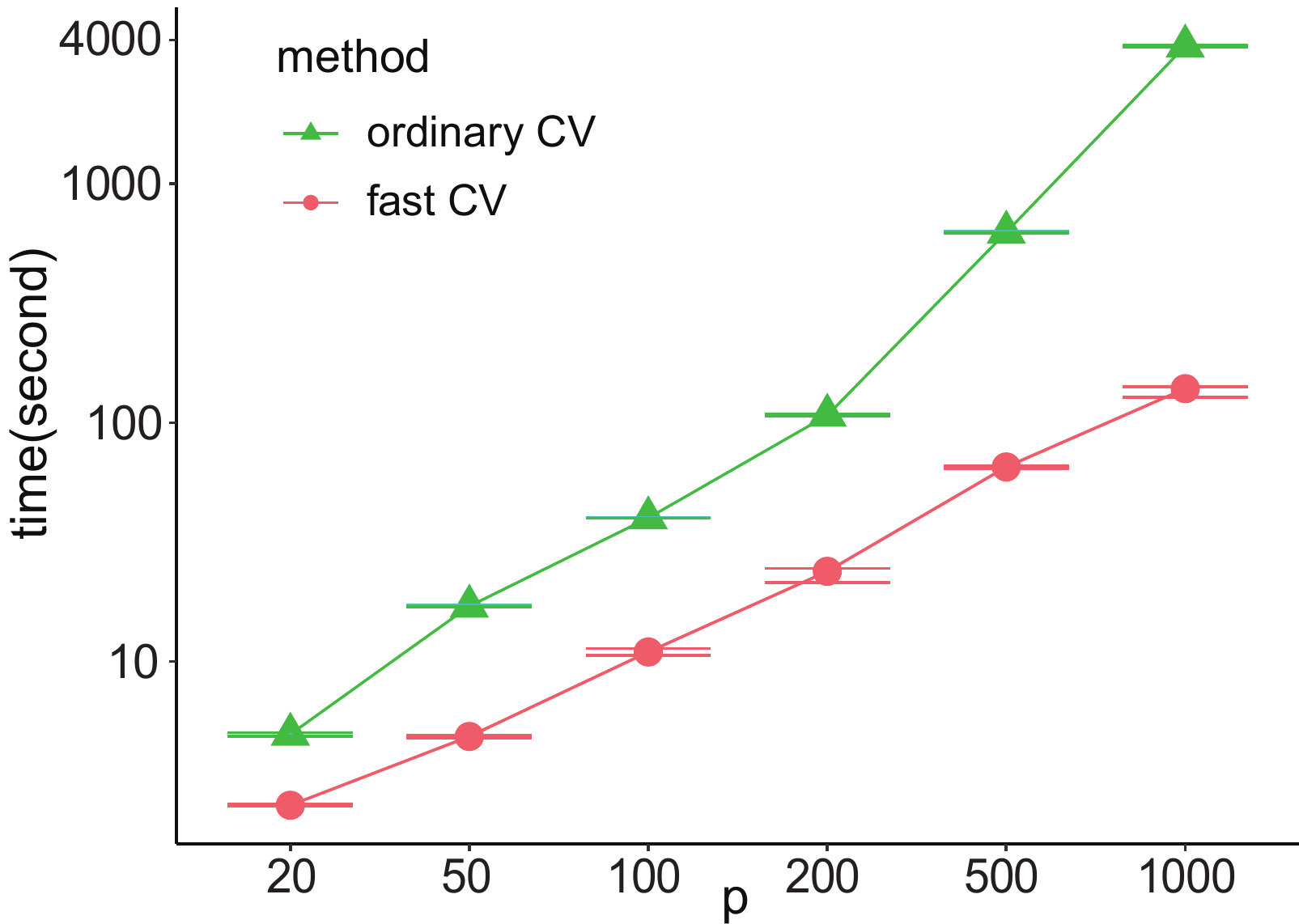}\label{fig:timing_err_bar}}
       \hspace{1cm}
    \subfigure[Error rates as a function of $n$.]{\includegraphics[height=5cm,bb=0 0 535 328]{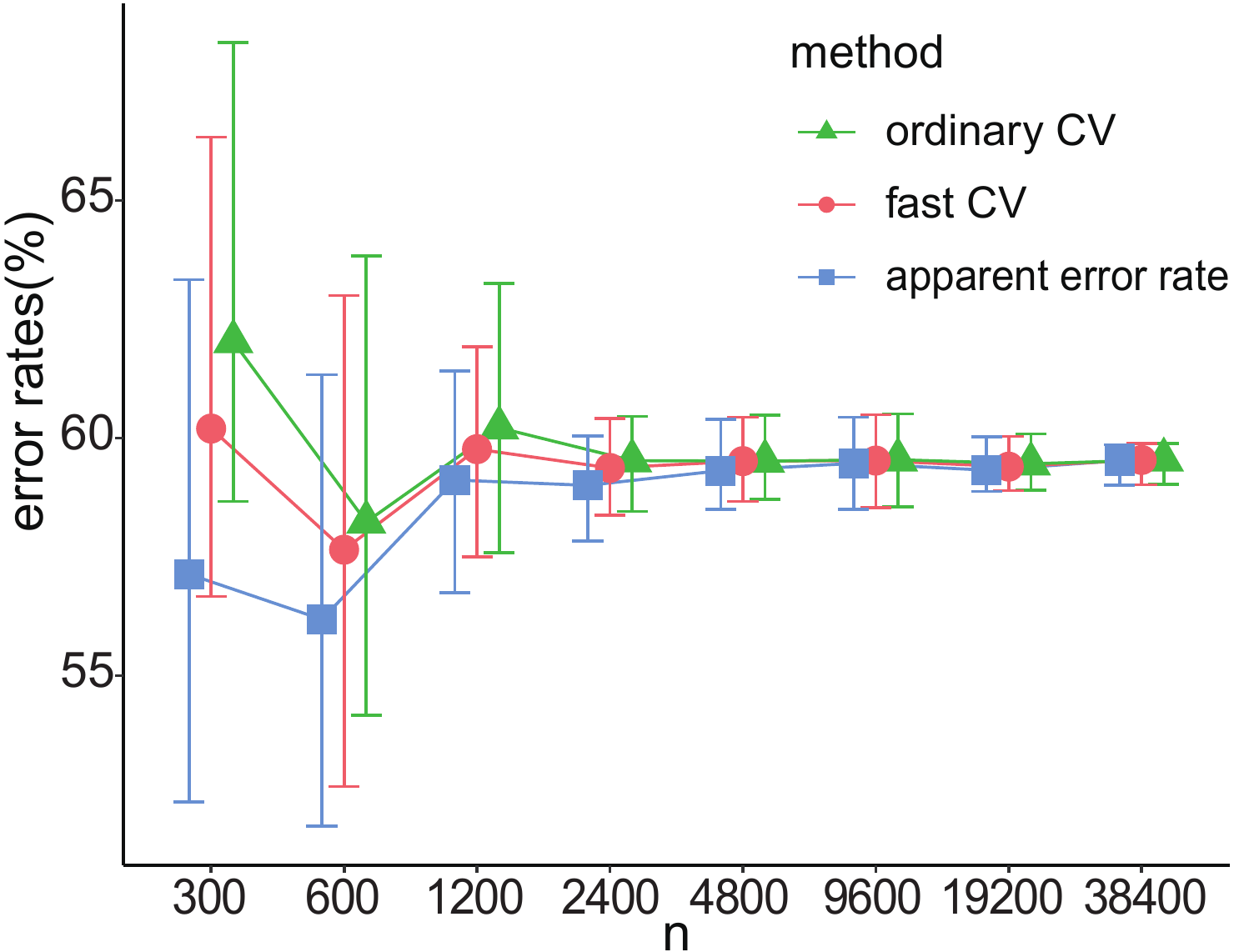}\label{fig:error_err_bar}}
		\caption{Comparison of fast CV and ordinary CV in terms of their (a) computational timings and (b) error rates. For comparison, we also compute apparent error rates (AER) in (b). The error bar represents the maximum and minimum values over 10 runs.} \label{fig:timingerror_err_bar}
	\end{center}
\end{figure}
The results show that our CV is generally faster than ordinary CV.  In particular, the difference between the timings of these two CV methods increases with $p$; therefore, our algorithm is efficient for high-dimensional data. The difference is caused by the fact that ordinary CV requires the computation of eigenvalues and eigenvectors of a $p \times p$ matrix for each observation.  

\subsubsection{Accuracy of fast CV}
The error rates of the two CV algorithms are also compared; these two CV algorithms result in slightly different values owing to some approximations of fast CV computation (see Section \ref{sec:construction of algorithm} for the details).  Note that the approximation of fast CV is justified for sufficiently large $n$, as described in Section \ref{sec:simsetup}; thus, we compare the CV values under various sample sizes, $n=300\times 2^{k-1} (k=1,\ldots,8)$, and examine whether fast CV numerically converges to ordinary CV as $n$ increases. The dimension of the predictor variable is $p=20$.   

Figure $\ref{fig:error_err_bar}$ depicts the error rates of two CV computations. For comparison, we also compute apparent error rates (AER). The error bar represents the maximum and minimum values over 10 runs. The results show that fast CV and AER underestimate true CV value because they use test data for training.  Fast CV provides a better approximation of true CV value than AER, especially for small and moderate sample sizes.  For example, when $n=2400$ (i.e., $n_j \approx 80$), the difference between two CV errors is not significant;  meanwhile, the AER may not provide a good approximation of true CV. On the other hand, for sufficient large $n$, three error rates have almost identical values.
%The width of error bars is so small that the proposed method almost improved the computational timing.  

\section{Real data analysis}
\label{sec:realdata}
We apply our method to mouse consomic strain data \citep{takada2008mouse}. The mouse consomic strain is obtained  by replacing a part of a chromosome of strain B6 with the MSM mouse's corresponding chromosome.  There are 30 classes $(J=30)$ and 36 features $(p=36)$, including body weight, body length, total cholesterol, and total protein.  The number of observations is 372; each class consists of 12 observations, on average.

The 2D data projection by LDA has already been shown in Figure \ref{fig:ldamouse}, suggesting that cluster analysis may improve the accuracy. Figure \ref{fig:error_step_realdata} shows the error rate as a function of the number of steps of our algorithm, $t$, when $D=2$.  
\begin{figure}[!t]
	\begin{center}
    \subfigure[Error rates as a function of the number of steps of the clustering algorithm, $t$, when $D=2$.]{\includegraphics[clip,height=6cm,bb=0 0 565 507]{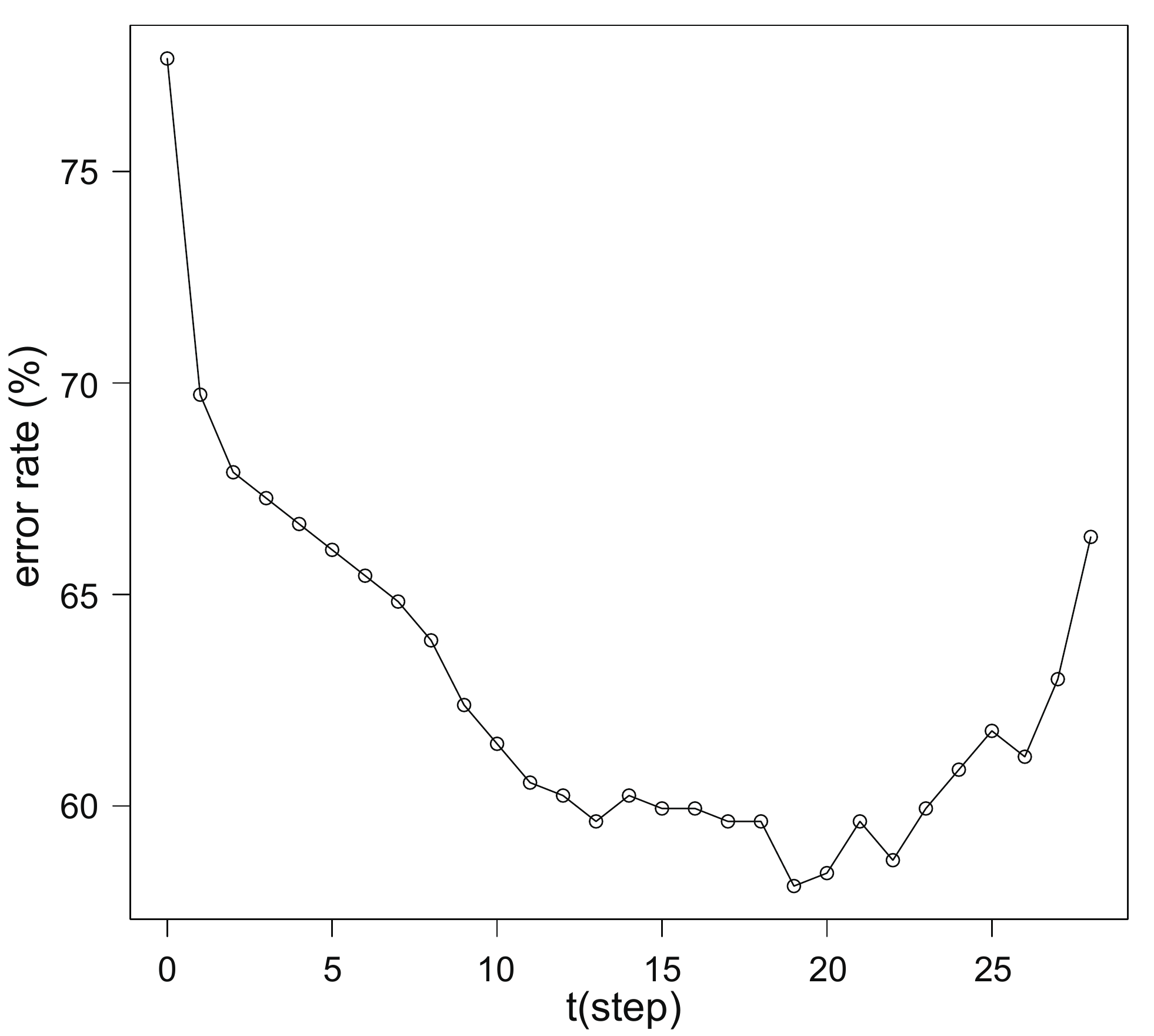}\label{fig:error_step_realdata}}
     \hspace{0.5cm}
    \subfigure[Error rates of HLDA, Ward's method, and ordinary LDA as a function of the dimension of the projected space, $D$.]{\includegraphics[clip,height=6cm,bb=0 0 639 504]{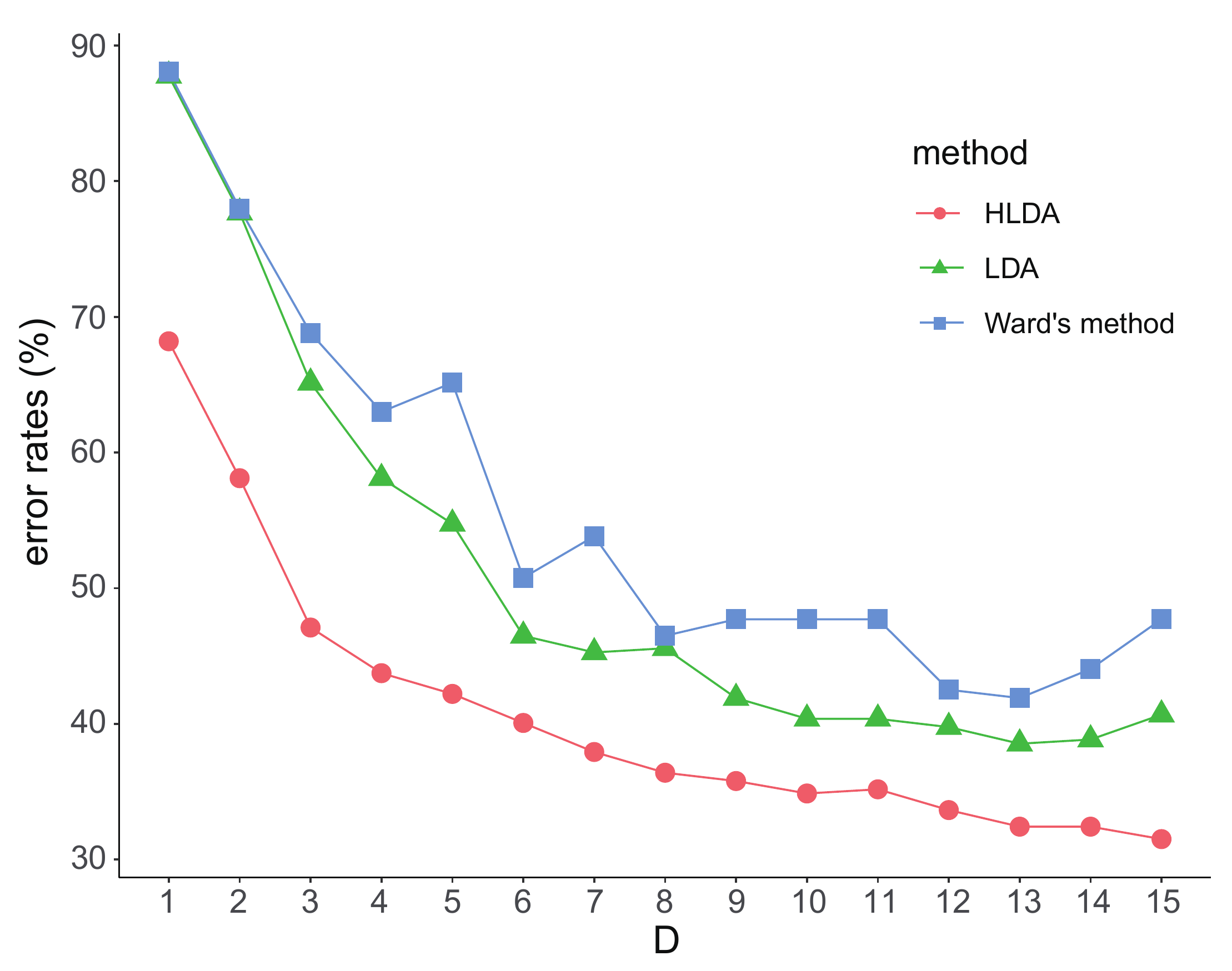}\label{fig:oresen_jitsu_err}}
		\caption{Error rates of real data analysis.}\label{fig:timingerror_err_bar}
	\end{center}
\end{figure}
The results show that our HLDA achieves higher accuracy than LDA for any $t$.  With an appropriate value of $t$, HLDA achieves an improvement of nearly 20\% in the error rate compared to LDA.  A similar tendency is found in the Monte Carlo simulation result for Model 2, as shown in Figure \ref{fig:hiera_errorbar_giji3}.   

The error rate as a function of $D$ is shown in Figure \ref{fig:oresen_jitsu_err}.  The results imply that HLDA outperforms existing methods for any $D$.   Ward's method performs worse than LDA, which suggests that it cannot achieve higher prediction accuracy than LDA for this dataset.  

Figure \ref{fig:visualization} shows the 2D projection of the data points obtained by HLDA. The number of clusters is $J=3$ for ease of interpretation of the clustering result.  The second cluster consists of only 20 observations with two classes, MSM and B6-9MSM; these two classes are separable with one dimension because $p=36$ $(> 20)$.   In particular, as shown in Figure \ref{fig:ldamouse_cluster}, MSM, which belongs to the second cluster, is far from the other classes.  Thus, the classification boundaries in the first and third clusters are not affected by MSM, allowing for reduction of the misclassification error rate.  Indeed, the CV values of HLDA and LDA are 63\% and 78\%, respectively, when $D=2$ and $J=3$.  
\begin{figure}[!t]
	\begin{center}
	 \subfigure[First 
	 cluster.]{\includegraphics[clip,height=8cm,bb=0 0 354 578]{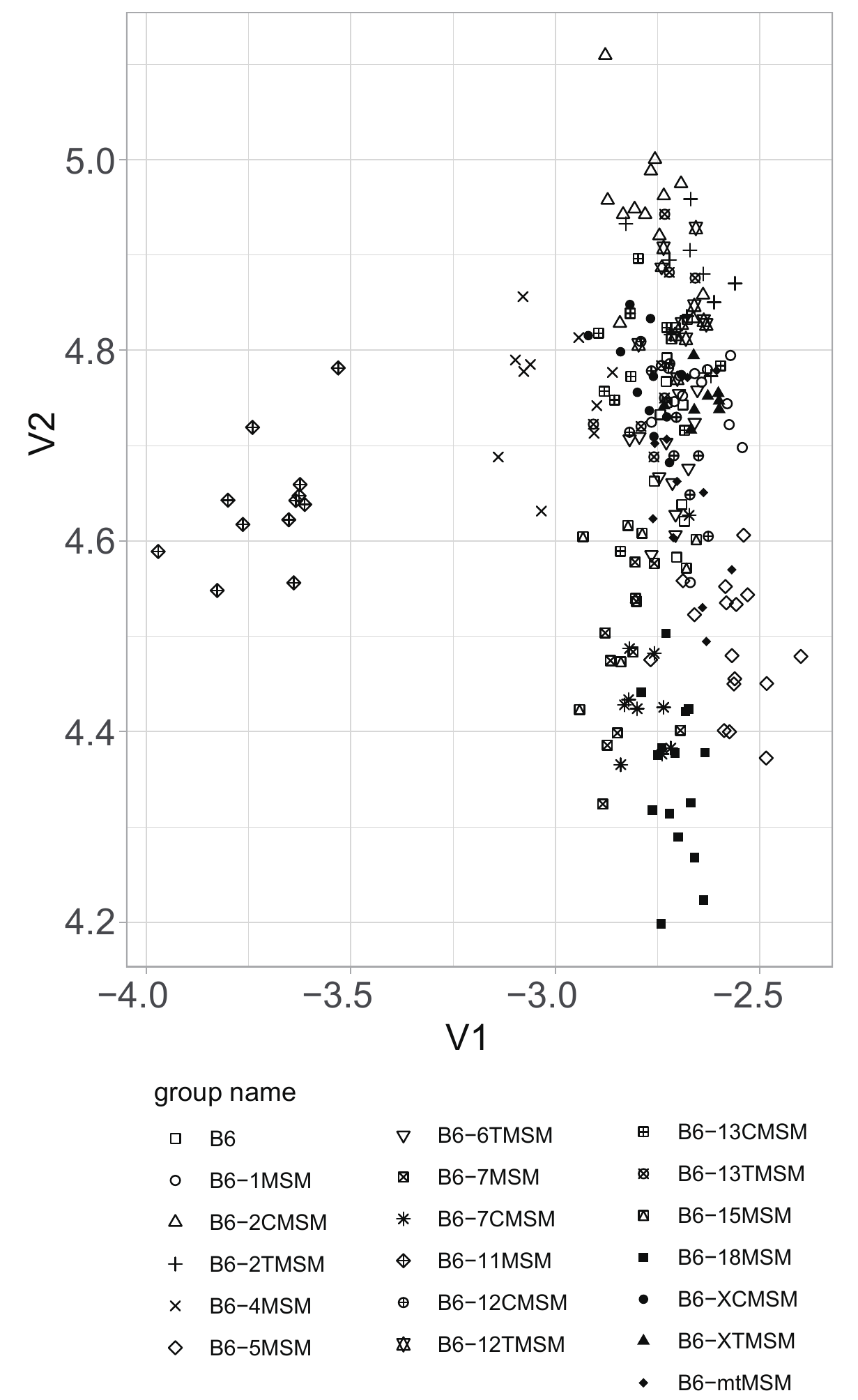}\label{fig:2dplot_realdata_1}}
    \hspace{0.2cm}    
\subfigure[Second cluster.]{\includegraphics[clip,height=8cm,bb=0 0 354 578]{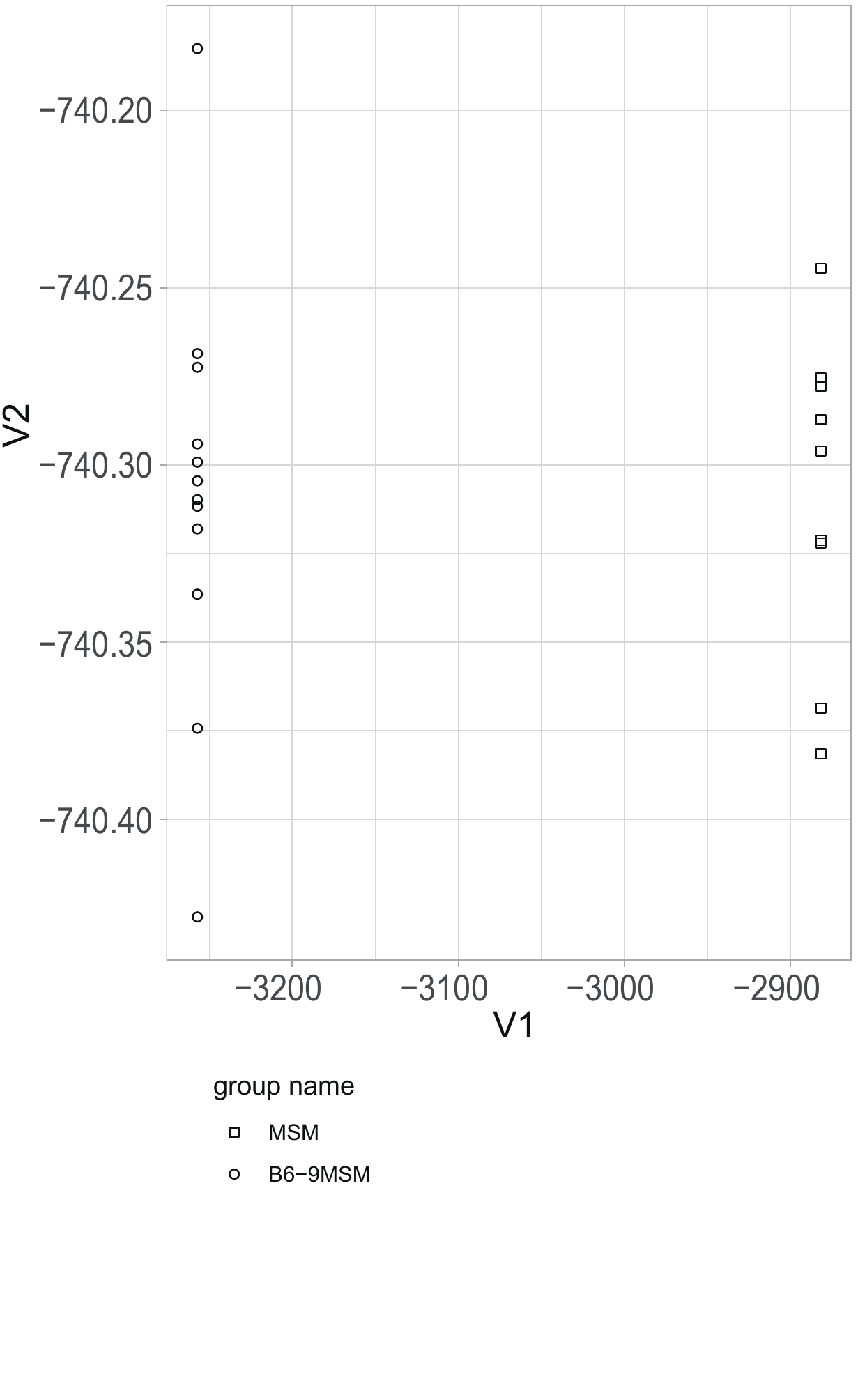}\label{fig:2dplot_realdata_1}}    
     \hspace{0.5cm}
     \subfigure[Third cluster.]{\includegraphics[clip,height=8cm,bb=0 0 354 578]{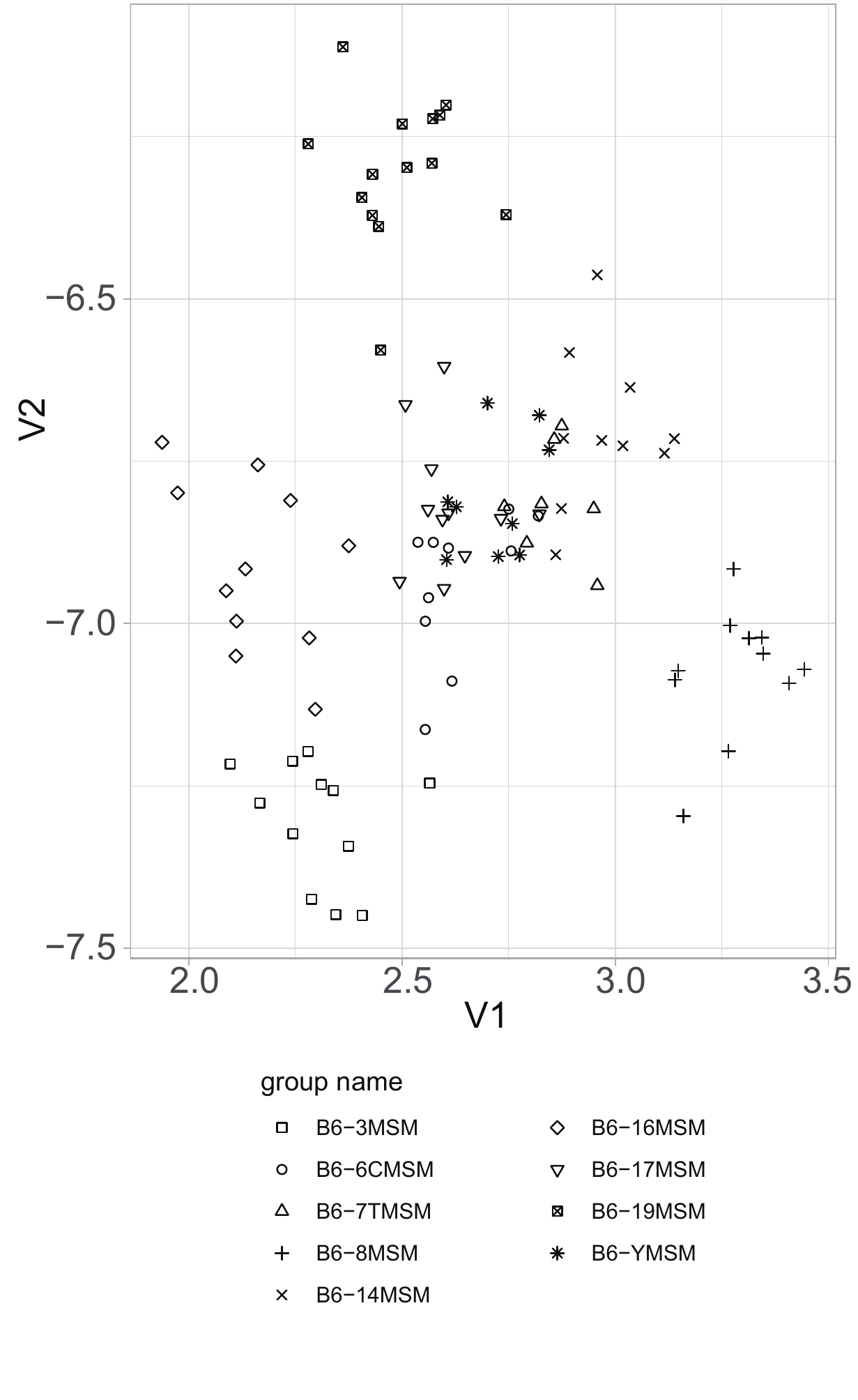}\label{fig:2dplot_realdata_1}}
	\end{center}

		\caption{2D plots for two clusters when $J=3$. } \label{fig:visualization}
\end{figure}

\section{Conclusion}
\label{sec:conclusion}
We considered a multiclass classification problem based on dimensionality reduction. In particular, we mainly dealt with a dataset that consists of a relatively large number of classes; for example, in real data analysis, we analyzed consomic strain data consisting of 30 classes. Through numerical experiments, we showed that our method achieves higher prediction accuracy than ordinary LDA while maintaining its advantages, including visualization.  

Our approach is based on hierarchical clustering; therefore, we cannot apply it to a dataset that consists of a massive number of classes (e.g., tens of thousands of classes) owing the prohibitively heavy computational load. As an alternative to hierarchical clustering, non-hierarchical clustering, such as $k$-means clustering or convex clustering \citep{chen2015convex,donnat2019convex}, can be applied to large-scale datasets. In the future, it would be interesting to extend our method to a non-hierarchical clustering technique.
\appendix
\section*{Appendix}
\section{Proof of Proposition \ref{prop:eigenreg}}\label{app:prop:eigenreg}

Recall that we have considered the following eigenvalue problem:
\begin{eqnarray}
\bm{S}_{W,\delta}^{-1/2}\bm{S}_B \bm{S}_{W,\delta}^{-1/2} \bm{s}_d = \lambda_d \bm{s}_d. \label{eq:eigen}
\end{eqnarray}
Because $\bm{S}_{B}$ is expressed as
\begin{eqnarray*}
\bm{S}_{B}=\frac{1}{2n^2}\sum_{j=1}^{J}\sum_{k \neq j}^{J} n_{j}n_k (\bar{\bm{x}}_{j}-\bar{\bm{x}}_k)(\bar{\bm{x}}_{j}-\bar{\bm{x}}_k)^{\mathsf{T}},
\end{eqnarray*}
$ \bm{s}_d$ is written as a linear combination of $\bm{S}_{W,\delta}^{-1/2} (\bar{\bm{x}}_{j}-\bar{\bm{x}}_k)$ $(j,k=1,\dots,J)$, i.e.,
\begin{equation*}
	 \bm{s}_d = \sum_{j=1}^{J}\sum_{k\neq j}^{J}\gamma_{jk,d} \bm{S}_{W,\delta}^{-1/2} (\bar{\bm{x}}_{j}-\bar{\bm{x}}_k), 
\end{equation*}
where $\gamma_{jk,d}$ $(j,k=1,\cdots,J)$ are scalar. Let $\bm{a}_{jk} = \bm{S}_{W,\delta}^{-1/2} (\bar{\bm{x}}_{j}-\bar{\bm{x}}_k)$. Eq. \eqref{eq:eigen} gives
\begin{equation*}
\bm{s}_d =\frac{1}{2n^2\lambda_d}  \sum_{j=1}^{J}\sum_{k\neq j}^{J} 	n_{j}n_k (\bm{a}_{jk}^{\mathsf{T}} \bm{s}_d )\bm{a}_{jk}.
\end{equation*}
Therefore, $\gamma_{jk,d}$ is expressed as
\begin{equation*}
	\gamma_{jk,d} = \frac{1}{2n^2\lambda_d}	n_{j}n_k (\bm{a}_{jk}^{\mathsf{T}} \bm{s}_d ). \label{eq:gamma_jk}
\end{equation*}

Then, $\bm{s}_d$ is written as  
\begin{eqnarray*}
	\bm{s}_d &=&  \bm{S}_{W,\delta}^{-1/2}\sum_{j=1}^{J}\sum_{k\neq j}^{J}\gamma_{jk,d}  (\bar{\bm{x}}_{j}-\bar{\bm{x}}_k) \\
	&=&  \bm{S}_{W,\delta}^{-1/2}\sum_{j=1}^{J}\sum_{k\neq j}^{J}(\gamma_{jk,d} - \gamma_{kj,d}) \bar{\bm{x}}_{j} \\
	&=&  \bm{S}_{W,\delta}^{-1/2}\sum_{j=1}^{J}\sum_{k\neq j}^{J}(2\gamma_{jk,d}) \bar{\bm{x}}_{j}.
%	&=&  \bm{S}_{W,\delta}^{-1/2}\sum_{j=1}^{g}\xi_j\bar{\bm{x}}_{j} \label{eq:s_xbar}
% &=&  \bm{S}_{W,\delta}^{-1/2}\sum_{j=1}^{g}\left(\sum_{k>j}^{g}\gamma_{jk} - \sum_{k<j}^{g}\gamma_{jk}\right)\bar{\bm{x}}_{j} \\
%	&=& \bm{S}_{W,\delta}^{-1/2}\sum_{j=1}^{g}\left(\sum_{k \neq j}^{g}{\rm sgn}(k-j)\gamma_{jk} \right)\bar{\bm{x}}_{j} \\
%	&=&\bm{S}_{W,\delta}^{-1/2} \sum_{j=1}^J \xi_j \bar{\bm{x}}_j
\end{eqnarray*}
%where $\xi_j = 2\sum_{k\neq j}^{g}\gamma_{jk}$.  
Here, we denote 
$$\xi_j = \frac{2n}{n_j}\sum_{k\neq j}^{J}\gamma_{jk,d} = \frac{1}{n\lambda_d} \sum_{k=1}^J n_k  (\bar{\bm{x}}_{j}-\bar{\bm{x}}_k)^{\mathsf{T}}\bm{S}_{W,\delta}^{-1/2} \bm{s}_d=\frac{1}{\lambda_d} (\bar{\bm{x}}_{j}-\bar{\bm{x}})^{\mathsf{T}}\bm{S}_{W,\delta}^{-1/2} \bm{s}_d.$$  
Then, we have
\begin{eqnarray}
	\bm{s}_d =  \bm{S}_{W,\delta}^{-1/2}\sum_{j=1}^{J}\frac{\xi_j}{n} n_j\bar{\bm{x}}_{j}. \label{eq:s_xbar}
\end{eqnarray}
From Eq. \eqref{eq:s_xbar} and the fact that $\sum_{j=1}^Jn_j\xi_j=0$, the proof is complete.
%\begin{eqnarray}
%	 \sum_{j=1}^J \bm{X}_j^{\mathsf{T}}\bm{y}_j &=& \bm{S}_{W,\delta}^{1/2}\bm{s}_d,\\
%	 \sum_{j=1}^J \bm{1}_{n_j}^{\mathsf{T}}\bm{y}_j &=& 0
%\end{eqnarray}

\section{Proof of Theorem \ref{thm:betat}}\label{app:thm:betat}
From Eq. \eqref{eq:n2}, $\hat{\beta}_{0d}$ satisfies
\begin{equation}
	\hat{\beta}_{0d} = -\frac{1}{n} \sum_{j=1}^Jn_j\bar{\bm{x}}_j^{\mathsf{T}}\hat{\bm{\beta}}_d \label{eq:beta0}
\end{equation}
Substituting Eq. \eqref{eq:beta0} into Eq. \eqref{eq:n1} yields
\begin{equation*}
	\left(\sum_{j=1}^J\bm{X}_j^{\mathsf{T}}\bm{X}_j + \delta\bm{I}_p \right) \hat{\bm{\beta}}_d -\frac{1}{n}  \sum_{j=1}^J \sum_{k=1}^Jn_jn_k\bar{\bm{x}}_j\bar{\bm{x}}_k^{\mathsf{T}}\hat{\bm{\beta}}_d = n\bm{S}_{W,\delta}^{1/2} \bm{s}_d
\end{equation*}
Since 
\begin{eqnarray*}
n\bm{S}_{W}&=&\sum_{j=1}^J\bm{X}_j^{\mathsf{T}}\bm{X}_j  - \sum_{j=1}^Jn_j\bar{\bm{x}}_j\bar{\bm{x}}_j^{\mathsf{T}},\\
n\bm{S}_{B}&=& \sum_{j=1}^{J}n_j\bar{\bm{x}}_j\bar{\bm{x}}_j^{\mathsf{T}} - \frac{1}{n} \sum_{j=1}^J \sum_{k=1}^Jn_jn_k\bar{\bm{x}}_j\bar{\bm{x}}_k^{\mathsf{T}},
\end{eqnarray*}
we have 
\begin{equation*}
	(\bm{S}_{W,\delta} + \bm{S}_B)\hat{\bm{\beta}}_d =  \bm{S}_{W,\delta}^{1/2}\bm{s}_d.
\end{equation*}
A simple calculation gives us
\begin{equation}
	(\bm{I} + \bm{S}_{W,\delta}^{-1/2}\bm{S}_B \bm{S}_{W,\delta}^{-1/2}) \bm{S}_{W,\delta}^{1/2}\hat{\bm{\beta}}_d = \bm{s}_d. \label{eq:alphaS}
\end{equation}
Meanwhile, from Eq. \eqref{eq:eigen}, $\bm{s}_d$ is expressed as
\begin{eqnarray}
\bm{s}_d =  \frac{1}{1+\lambda_d}(\bm{I} + \bm{S}_{W,\delta}^{-1/2}\bm{S}_B \bm{S}_{W,\delta}^{-1/2})\bm{s}_d.  \label{eq:eigenSw}
\end{eqnarray}
Substituting Eq. \eqref{eq:eigenSw} into Eq. \eqref{eq:alphaS} yields
\begin{equation*}
	(\bm{I} + \bm{S}_{W,\delta}^{-1/2}\bm{S}_B \bm{S}_{W,\delta}^{-1/2}) \bm{S}_{W,\delta}^{1/2} \left(\hat{\bm{\beta}}_d - \frac{1}{1+\lambda_d} \bm{S}_{W,\delta}^{-1/2} \bm{s}_d \right) = \bm{0}. 
\end{equation*}
Because the inverse matrix of $(\bm{I} + \bm{S}_{W,\delta}^{-1/2}\bm{S}_B \bm{S}_{W,\delta}^{-1/2}) \bm{S}_{W,\delta}^{1/2}$ always exists, the proof is complete.

\section{Proof of Theorem \ref{thm:x^Talpha}}\label{app:thm:x^Talpha}
	Let $\bm{C} = \tilde{\bm{X}}^{\mathsf{T}}\tilde{\bm{X}} + \bm{\Delta}$ and $\bm{C}^{(-i)} = (\tilde{\bm{X}}^{(-i)})^{\mathsf{T}}\tilde{\bm{X}}^{(-i)} + \bm{\Delta}$. Because $\bm{C} = \bm{C}^{(-i)}  + \bm{x}_i\bm{x}_i^{\mathsf{T}}$, the inverse of $\bm{C}^{(-i)}$ is calculated using Woodbury's formula (e.g., \citealp{BA8658116X}):
	\begin{equation}
	(\bm{C}^{(-i)})^{-1} = \bm{C}^{-1} + \dfrac{\bm{C}^{-1} \tilde{\bm{x}}_i\tilde{\bm{x}}_i^{\mathsf{T}} \bm{C}^{-1}}{1 - \tilde{\bm{x}}_i^{\mathsf{T}} \bm{C}^{-1}  \tilde{\bm{x}}_i}. \label{eq:woodbury}
\end{equation}
Substituting Eq. \eqref{eq:woodbury} into Eq. \eqref{eq:alpha_-i} yields
\begin{eqnarray*}
	\hat{\bm{\alpha}}_d^{*(-i)} &=& \left(\bm{C}^{-1} + \frac{ \bm{C}^{-1}\tilde{\bm{x}}_i\tilde{\bm{x}}_i^{\mathsf{T}}\bm{C}^{-1}}{1-h_{ii}}\right)(\tilde{\bm{X}}^{\mathsf{T}}\bm{y}_d - y_{id}\tilde{\bm{x}}_i)\\
           &=& \hat{\bm{\alpha}}_d^* - y_{id}\bm{c}_i + \frac{\hat{y}_{id}\bm{c}_i - y_{id}h_{ii}\bm{c}_i}{1-h_{ii}}\\
           &=& \hat{\bm{\alpha}}_d^* + \frac{\hat{y}_{id}-y_{id}}{1-h_{ii}}\bm{c}_i,
\end{eqnarray*}
where $\bm{c}_i = \bm{C}^{-1}\tilde{\bm{x}}_i$.  Therefore, Eq. \eqref{eq:ialpha} holds.  Eq. \eqref{eq:xialpha} immediately holds using Eq. \eqref{eq:ialpha}.

\section{Proof of Lemma \ref{lemma:x^a_eigenvalue}}\label{proof:lemma:x^a_eigenvalue}
%To show Lemma \ref{lemma:x^a_eigenvalue}, it is sufficient to show that
%\begin{eqnarray*}
%	\sum_{i=1}^n (\tilde{\bm{x}}_i^{\mathsf{T}} \hat{\bm{\alpha}}_d)^2 = \frac{1}{1+\lambda_d}.
%\end{eqnarray*}
From Eq. \eqref{eq:trans_regress}, we have
\begin{eqnarray*}
	\frac{1}{n}\sum_{i=1}^n (\tilde{\bm{x}}_i^{\mathsf{T}} \hat{\bm{\alpha}}_d)^2 + \delta \hat{\bm{\beta}}_d^{\mathsf{T}}\hat{\bm{\beta}}_d &=& \frac{1}{(1+\lambda_d)^2}\bm{t}_{d,\delta}^{\mathsf{T}} \left\{ \frac{1}{n} \sum_{i=1}^n  (\bm{x}_i - \bar{\bm{x}}) (\bm{x}_i - \bar{\bm{x}})^{\mathsf{T}} + \delta \bm{I}\right\} \bm{t}_{d,\delta} \\
	&=& \frac{1}{(1+\lambda_d)^2}\bm{t}_{d,\delta}^{\mathsf{T}}(\bm{S}_{W,\delta} + \bm{S}_B)\bm{t}_{d,\delta}\\
	&=& \frac{1}{(1+\lambda_d)^2}\bm{s}_d^{\mathsf{T}}(\bm{I} + \bm{S}_{W,\delta}^{-1/2}\bm{S}_B\bm{S}_{W,\delta}^{-1/2})\bm{s}_d \\
	&=& \frac{1}{(1+\lambda_d)^2} (1+\lambda_d) \\
	&=& \frac{1}{1+\lambda_d}.
\end{eqnarray*}
The proof is complete.

\section{Proof of Proposition \ref{prop:consistency}}\label{sec:proofconsistency}
It is sufficient to show that $\xi_{jd} $ converges almost surely to some constant value because $\bm{y}_{jd}$ is defined as $\bm{y}_{jd} = \xi_{jd}\bm{1}_{n_j}$.  Under Assumption \ref{assumption},  $\bm{S}_{B,\delta}$ and $\bm{S}_{W,\delta}$ converge almost surely to some constant matrices, say $\bm{\Sigma}_{B,\delta}$ and $\bm{\Sigma}_{W,\delta}$, respectively. Let $(\lambda_{0,d},\bm{s}_{0,d})$ be the $d$th largest eigenvalue and eigenvectors of $\bm{\Sigma}_{W,\delta}^{-1/2}\bm{\Sigma}_{B,\delta}\bm{\Sigma}_{W,\delta}^{-1/2}$, respectively.  Since $\xi_{jd} $ in Eq. \eqref{eq:responses} is expressed as 
\begin{eqnarray*}
	\xi_{jd} = \frac{1}{n\lambda_d} \sum_{k=1}^J n_k  (\bar{\bm{x}}_{j}-\bar{\bm{x}}_k)^{\mathsf{T}}\bm{S}_{W,\delta}^{-1/2} \bm{s}_d  = \frac{1}{\lambda_d} (\bar{\bm{x}}_j-\bar{\bm{x}})^{\mathsf{T}}\bm{S}_{W,\delta}^{-1/2}\bm{s}_d \quad (j=1,\dots,J),
\end{eqnarray*}
we have
\begin{eqnarray*}
	\xi_{jd} \as  \xi_{0,jd}:=\frac{1}{\lambda_{0,d}}  \left({\bm{\mu}}_j-\sum_{j=1}^Jc_j\bm{\mu}_j\right)^{\mathsf{T}}\bm{\Sigma}_{W,\delta}^{-1/2}\bm{s}_{0,d}  \quad (j=1,\dots,J).
	\end{eqnarray*}
Thus, Eq. \eqref{eq:propconsistency1} holds.  It is also shown that the estimators $\hat{\bm{\alpha}}_d^{(-i)}$ and $\hat{\bm{\alpha}}_d^{*(-i)}$ converge to the same vector:
\begin{eqnarray}
		\hat{\bm{\alpha}}_d^{(-i)} 
		%&=& \left\{(\tilde{\bm{X}}^{(-i)})^{\mathsf{T}}\tilde{\bm{X}}^{(-i)} + \bm{\Delta}\right\}^{-1}(\tilde{\bm{X}}^{(-i)})^{\mathsf{T}}\bm{y}_d^{(-i)} {\xrightarrow{\rm a.s.}} 
		\as \bm{\alpha}_d, \quad  
		\hat{\bm{\alpha}}_d^{*(-i)} 
		%&=& \left\{(\tilde{\bm{X}}^{(-i)})^{\mathsf{T}}\tilde{\bm{X}}^{(-i)} + \bm{\Delta}\right\}^{-1}(\tilde{\bm{X}}^{(-i)})^{\mathsf{T}}\bm{y}_d^{*(-i)} {\xrightarrow{\rm a.s.}} 
		\as \bm{\alpha}_d, \quad  \bm{\alpha}_d:= \left(\sum_j \bm{A}_j \right)^{-1}\sum_{j=1}^J \xi_{0,jd} \bm{\mu}_j. \label{eq:alphaconsistency}
\end{eqnarray}
Eq. \eqref{eq:lemma:x^a_eigenvalue0} implies that
\begin{eqnarray}
	\bm{\alpha}_d^{\mathsf{T}}\left(\sum_{j=1}^Jc_j\bm{A}_j - \bm{\mu}_j\bm{\mu}_j^{\mathsf{T}}\right)\bm{\alpha}_d + \delta \bm{\beta}_d^{\mathsf{T}}\bm{\beta}_d = \frac{1}{1+\lambda_{0,d}}.   \label{eq:lamtrueeq}
\end{eqnarray}	
Combining Eqs. \eqref{eq:alphaconsistency} and \eqref{eq:lamtrueeq}, $\lambda_d^{*(-i)}$ defined in Eq. \eqref{eq:lambdaast_def} converges to the true eigenvalue, i.e.,
\begin{eqnarray*}
	\lambda_d^{*(-i)} \as \lambda_{0,d}.
\end{eqnarray*}
The proof is complete using the continuous mapping theorem.
%\begin{acknowledgements}
%If you'd like to thank anyone, place your comments here
%and remove the percent signs.
%\end{acknowledgements}

% BibTeX users please use one of
%\bibliographystyle{spbasic}      % basic style, author-year citations
%\bibliographystyle{spmpsci}      % mathematics and physical sciences
%\bibliographystyle{spphys}       % APS-like style for physics
%\bibliography{}   % name your BibTeX data base

% Non-BibTeX users please use

\bibliographystyle{abbrvnat}
\bibliography{papers}
\end{document}